\documentclass[letterpaper,11pt]{article}
\usepackage{amsmath}
\usepackage{amssymb}
\usepackage{amsthm}
\usepackage{amsfonts}
\usepackage{setspace}
\usepackage[bottom]{footmisc}
\usepackage{bbm}
\usepackage{color}
\usepackage{enumerate}
\usepackage{mathrsfs}
\usepackage{graphicx}
\usepackage{subcaption}
\usepackage{hyperref}
\usepackage{enumitem}
\usepackage[usenames,dvipsnames,svgnames,table]{xcolor}
\usepackage{multirow}
\usepackage[margin=1.25in]{geometry}
\usepackage{tikz}
\usetikzlibrary{fit,positioning,arrows,automata}
\usepackage{enumitem}
\setlist{nosep}

\newtheorem{theorem}{Theorem}[section]
\newtheorem{lemma}{Lemma}[section]
\newtheorem{proposition}{Proposition}[section]

\newtheorem{corollary}{Corollary}[section]

\newtheorem{remark}{Remark}[section]
\makeatletter
\def\thm@space@setup{%
  \thm@preskip=16pt
  \thm@postskip=5pt 
}
\makeatother

\newtheoremstyle{theoremd}
  {\topsep}
  {\topsep}
  {\itshape}
  {0pt}
  {\bfseries}
  {}
  { }
  {\thmname{#1}\thmnumber{ }\thmnote{ (#3)}}

\theoremstyle{theoremd}

\newcommand{\mb}[1]{\mathbb{#1}}

\newcommand{\mr}[1]{\mathrm{#1}}
\newcommand{\mcr}[1]{\mathscr{#1}}
\newcommand{\mc}[1]{\mathcal{#1}}

\newcommand{\ind}{1\!\mathrm{l}}

\newcommand{\ul}[1]{\underline{#1}}

\hypersetup{
     colorlinks   = true,
     citecolor    = Blue,
     urlcolor     = Black,
     linkcolor    = Blue
}

\makeatletter
\renewcommand\paragraph{\@startsection{paragraph}{4}{\z@}%
                                    {0pt \@plus1ex \@minus.2ex}%
                                    {-1em}%
                                    {\normalfont\normalsize\bfseries}}
\makeatother

\usepackage{bibunits}
\usepackage[longnamesfirst]{natbib}

\begin{document}

\defaultbibliography{tp}
\defaultbibliographystyle{chicago}
\begin{bibunit}

\author{%
{Timothy M. Christensen\thanks{%
Department of Economics, New York University, 19 W. 4th Street, 6th floor, New York, NY 10012, USA. E-mail address: \texttt{timothy.christensen@nyu.edu}}
}
} 

\title{%
Existence and uniqueness of recursive utilities \\without boundedness\footnote{
This paper is a revised version of the preprint arXiv:1812.11246 which was posted to arXiv in December 2018. I thank J. Borovi\v{c}ka, L. Hansen, T. Sargent, and participants of the ``Blue Collar Working Group'' at the University of Chicago, and several seminar audiences for helpful comments. This material is based upon work supported by the National Science Foundation under Grant No. SES-1919034.}
}

\date{August 17, 2021. }

\maketitle

\begin{abstract}  

\noindent
This paper derives primitive, easily verifiable sufficient conditions for existence and uniqueness of (stochastic) recursive utilities for several important classes of preferences. In order to accommodate models commonly used in practice, we allow both the state space and per-period utilities to be unbounded. For many of the models we study, existence and uniqueness is established under a single, primitive ``thin tail'' condition on the distribution of growth in per-period utilities. We present several applications to robust preferences, models of ambiguity aversion and learning about hidden states, and Epstein--Zin preferences.

\medskip 


\medskip


\end{abstract}

\pagenumbering{arabic}

\newpage

\section{Introduction}

Recursive utilities\footnote{Throughout the paper, by ``recursive utility'' we mean ``stochastic recursive utility''.} play a central role in contemporary macroeconomics and finance. Under recursive preferences, the value of a stream of per-period utilities is  defined as the solution to a nonlinear, stochastic, forward-looking difference equation (or ``recursion''). Despite the importance of recursive utilities, existence and uniqueness remains an unresolved issue as the recursions are typically not contraction mappings when state variables and per-period utilities are unbounded. In this paper, we derive primitive, easily verifiable sufficient conditions for existence and uniqueness of recursive utilities in stationary, infinite-horizon Markovian environments, with an emphasis on robust preferences, models of ambiguity aversion and learning about hidden states, and Epstein--Zin preferences. To accommodate parameterizations of models used extensively in macroeconomics and finance, we allow both the support of the Markov state vector and per-period utilities to be unbounded. 

There are a large number of existence and uniqueness results for recursive utilities in models with compact state space, and possibly also bounded per-period utilities.\footnote{See, e.g., \cite{EZ}, \cite{AlvarezJermann2005}, \cite{MarinacciMontrucchio2010}, \cite{GuoHe}, \cite{BeckerRinconZapatero}, \cite{BloiseVailakis}, \cite{Balbus},  \cite{BorovickaStachurski}, \cite{RenStachurski}, and references therein.} However, many models used in macroeconomics and finance feature unbounded (i.e., non-compact) state spaces and unbounded utilities. For instance, the extensive long-run risks literature following \cite{BansalYaron2004} typically models state variables  as vector autoregressive processes with unbounded shocks.\footnote{See, e.g., \cite{HHL}, \cite{BHS}, \cite{Wachter}, \cite{BKSY2014}, \cite{CLL}, \cite{BS}, \cite{CMST}, \cite{SSY}, and additional references listed in Sections \ref{s:robust}--\ref{s:ez}.} A seemingly reasonable approach for models with non-compact state space is to truncate (i.e., compactifty) the state space and apply existing results for compact state spaces. After all, this truncation occurs implicitly when computing solutions numerically. However truncation, even at an arbitrarily high truncation level, can materially alter the existence and uniqueness properties of the recursions we study. Knowing when the original model without truncation has a unique solution remains important for reconciling numerical solutions with the original (un-truncated) model envisioned by the researcher.

To illustrate this point, in Section \ref{sec:problems} we present two empirically relevant examples to show how non-existence and non-uniqueness can arise under unboundedness. For both examples, we focus on a recursion arising under preferences for ``robustness'' \citep{HS1995,HS2001,HSTW} and under Epstein--Zin preferences with unit intertemporal elasticity of substitution. The first example is a simplified version of the consumption growth process from \cite{SSY}, for which existence fails. The second example is from \cite{BS} and \cite{Wachter}, for which uniqueness fails. When the state space is truncated, however, the recursion has a unique solution for both examples (irrespective of the truncation level). This stark difference between the compact and unbounded case arises because the properties of this recursion depend delicately on the tail behavior of state variables and truncation, even at an arbitrarily high truncation level, materially alters tail behavior.

For many of the models we study, the single primitive sufficient condition for both existence \emph{and} uniqueness is that the distribution of growth in per-period utilities has thin tails, in a sense we make precise.  We verify this condition for robust preferences, models of ambiguity aversion and learning about hidden states, and Epstein--Zin preferences with unit intertemporal elasticity of substitution (IES). We consider both canonical linear-Gaussian environments which are pertinent to the long-run risks literature as well as environments featuring regime-switching and stochastic volatility.

As with much of the literature, we identify recursive utilities with fixed points of a nonlinear operator acting on a suitable function class. There exists a literature on existence and uniqueness for (deterministic or stochastic) utilities under unboundedness using contraction mapping arguments for function classes defined via weighted sup-norms.\footnote{See, e.g., \cite{Boyd} and \cite{Duran} for deterministic and stochastic utilities, respectively. \cite{LeVanVailakis} provide a related approach for deterministic utilities under Lipschitz conditions.} However, it is not always easy to find a suitable weighting function under which operators defining recursive utilities are a contraction.\footnote{See, e.g., \cite{LeVanVailakis} for a discussion.} Our arguments instead rely on monotonicity and concavity/convexity properties of the recursions we study, as with earlier work by \cite{MarinacciMontrucchio2010}; see also \cite{BeckerRinconZapatero}, \cite{BloiseVailakis}, and \cite{RenStachurski}, primarily for the compact case.\footnote{\cite{MarinacciMontrucchio2010} and \cite{BeckerRinconZapatero} allow for processes that are bounded with probability one but growing over time using weighted $\ell^\infty$-norms. See also \cite{RenStachurski} for a particular parameterization of Epstein--Zin preferences with unbounded state space using a weighted sup-norm, where the weighting function is tightly related to per-period utilities and the law of motion of the Markov state.} While our approach has some similarities with these earlier works, it differs  in terms of the function class and technical arguments used so as to accommodate a broad class of empirically relevant models with unbounded state space. In particular we do not rely on topological properties of the space of bounded functions, such as the such the solidness of the positive cone. 

Our point of departure is to embed a transformation of the value function, such as its logarithm, in a class of unbounded but thin-tailed functions. The class is an  exponential-Orlicz class used in empirical process theory in statistics \citep{vdVW} and modern high-dimensional probability \citep{Vershynin}.\footnote{Previously, \cite{HindyHuang} and \cite{HindyHuangKreps} used Orlicz classes to define topologies for consumption paths in continuous time.} Exponential-Orlicz classes are naturally suited to the recursions we study, which involve the composition of exponential and logarithmic transforms and expected values. 

The key high-level condition we use to establish uniqueness is that  a subgradient (in the convex case) or supergradient (in the concave case) of the recursion is monotone and its spectral radius is strictly less than one.
For many of the models we study, the recursion is convex and its subgradient is a discounted conditional expectation under a distorted law of motion. Verifying the spectral radius condition in these models amounts to checking a primitive thin-tail condition on the change-of-measure distorting the law of motion. We specialize this condition to particular models, deriving more primitive thin-tail conditions on the distribution of growth in per-period utility which are easy to verify: one simply has to know the tail behavior of the distribution. 

To illustrate the usefulness of our results, we then present applications to three classes of models. 

Section \ref{s:robust} studies a recursion arising under preferences for ``robustness'', namely risk-sensitive preferences \citep{HS1995}, multiplier preferences \citep{HS2001}, constraint preferences \citep{HSTW}, and also under \cite{EZ} preferences with unit IES. There are currently no uniqueness results in the literature for this recursion allowing non-compact state space and unbounded utilities (see the discussion in Section \ref{s:robust}), both of which are important for models in macroeconomics and finance. We establish new existence and uniqueness results under a single, primitive thin-tail condition on utility growth. We verify this condition in canonical linear-Gaussian environments and environments featuring regime-switching and stochastic volatility, thereby establishing new existence and uniqueness results for such settings.

Section \ref{s:learn} considers models with learning. We study extensions of multiplier preferences to accommodate both model uncertainty and uncertainty about hidden states due to \cite{HS2007,HS2010}, dynamic models of ambiguity aversion studied by \cite{JuMiao} and \cite{KMM2009}, and Epstein--Zin preferences with unit IES and learning. There are currently no existence and uniqueness results in the literature allowing non-compact state space and unbounded utilities (see the discussion in Section \ref{s:learn}). We establish existence and uniqueness under a single, primitive thin-tail condition on utility growth. We verify the condition, and therefore establish existence and uniqueness results, for regime-switching environments \citep{JuMiao} and Gaussian state-space models \citep{HS2007,HS2010,CLL,CMST}. 

Finally, in Section \ref{s:ez} we examine Epstein--Zin recursive utilities with IES not equal to one. There are no uniqueness results for models with unbounded state space when risk aversion and intertemporal substitution are in a range normally encountered in the long-run risks literature (see the discussion in Section \ref{s:ez}). Here we establish existence under an eigenvalue condition from \cite{HS2012} and a thin-tail condition on its corresponding eigenfunction. We verify this condition for linear-Gaussian environments which are pertinent to the long-run risks literature.
All proofs are in Appendix \ref{s:proofs}.

\section{Non-existence and non-uniqueness without boundedness}\label{sec:problems}

In this section, we present two empirically relevant examples of non-existence and non-uniqueness in models with unbounded state spaces. The first uses a simplified version of the state process from \cite{SSY}, for which existence fails. The second is the model from \cite{BS} and \cite{Wachter} for which uniqueness fails. In both examples, however, there is always a unique solution when the support of state variables are truncated (irrespective of the truncation level).

\subsection{Non-existence}

Consider the following simplified\footnote{We have removed a stochastic growth component for $g$ from model (4) of \cite{SSY} to simplify presentation. Non-uniqueness arises here because of the form of the stochastic volatility process, and not because of the absence of a stochastic growth component for $g$.} model of consumption growth $g$ from \cite{SSY}:
\begin{equation}
 g_{t+1}  = \nu_g + e^{h_t} \eta^g_{t+1} \,, \quad \quad
 h_{t+1}  = \nu_h + \rho h_t + \sigma \eta^h_{t+1} \,, \label{eq:ssy}
\end{equation}
where $|\rho| < 1$, and the $\eta^g_t$ and $\eta^h_t$ are all i.i.d. $N(0,1)$. Let $X_t = (g_t,h_t)$. The supports of $g_t$ and $h_t$ are both $\mb R$. 

Suppose we seek a solution $v$ to the recursion
\begin{equation} \label{eq:ssy_recur_0}
 v(x) = \beta \log \mb E^{Q} \left[ \left. e^{v(X_{t+1}) + \alpha g_{t+1}} \right| X_t = x \right] ,
\end{equation}
where $\beta \in (0,1)$ and $\alpha \in \mb R$ are preference parameters and $\mb E^Q$ denotes expectation under the law of motion (\ref{eq:ssy}). This recursion is studied in Section \ref{s:robust} and arises under various preferences for robustness as well as under Epstein--Zin preferences with unit IES. As the conditional distribution of $X_{t+1}$ given $X_t = (g,h)$ depends only on $h$, the right-hand side conditional expectation, and therefore $v$, must depend only on $h$. Using (\ref{eq:ssy}), we see that recursion (\ref{eq:ssy_recur_0}) simplifies to
\begin{equation} \label{eq:ssy_recur}
 v(h) = \mathsf a + \mathsf b e^{2h} + \beta \log \mb E^{Q} \left[ \left. e^{v(h_{t+1})} \right| h_t = h \right] =: \mb T v(h),
\end{equation}
where $\mathsf a = \alpha \beta \nu_g$ and $\mathsf b = \frac{1}{2} \alpha^2 \beta$.

Let $L^1$ denote the class of functions of $h$ with finite expectation under the stationary distribution $\mu$ of $h$.

\begin{proposition} \label{prop:ssy-nonexist}
Let $\alpha \neq 0$ and let consumption growth $g$ evolve according to (\ref{eq:ssy}). Then: recursion (\ref{eq:ssy_recur}) has no solution in $L^1$.
\end{proposition}

Now suppose instead that the support of $h$ is truncated to some compact interval $\mc H := [-H,H]$ for $H \in (0,\infty)$. Under this truncation, $\mb T$ satisfies Blackwell's sufficient conditions for a contraction mapping on the space $B(\mc H)$ of bounded functions on $\mc H$. Therefore, $\mb T$ has a unique fixed point in $B(\mc H)$, irrespective of the truncation level $H$.

To understand the difference between the bounded and unbounded cases, note from (\ref{eq:ssy_recur_0}) that we need the tails of the (conditional) distribution of $v(X_{t+1}) + \alpha g_{t+1}$ to decay sufficiently quickly for $\mb T v$ to be well defined. While this condition is always satisfied in the bounded case,
it is violated in model (\ref{eq:ssy}) due to the specification of the stochastic volatility process. 
In Section \ref{s:robust} we present a different form of stochastic volatility with thinner tails for which existence and uniqueness can be guaranteed without truncation. 

\subsection{Non-uniqueness}\label{sec:bs}

Consider the model from \cite{BS} (see also \cite{Wachter}) in which consumption growth $g$ evolves as
\begin{equation} \label{eq:bs}
 g_{t+1} = \nu_g + w_{z,t+1} + \sigma w_{g,t+1} \,,
\end{equation}
where $w_{g,t+1} \sim N(0,1)$ and $w_{z,t+1}|j_{t+1} \sim N(\nu_j j_{t+1}, \sigma_j^2 j_{t+1})$ with $\nu_j < 0$, and where $j_{t+1}|h_t$ is Poisson-distributed with mean $h_t$, where $h$ follows an autoregressive gamma process with parameters $(\varphi, c, \delta)$ (see appendix H of \cite{BCZ} and references therein for details). Here consumption growth is subject to occasional ``disasters'' which arrive at rate $h_t$. We again seek a solution to recursion (\ref{eq:ssy_recur_0}) with $X_t = (g_t,h_t)$. The support of $g_t$ is $\mb R$ and the support of $h_t$ is $\mb R_+$. As with the previous example, here it suffices to consider solutions depending only on $h$. Using (\ref{eq:bs}), we may rewrite recursion (\ref{eq:ssy_recur_0}) as
\begin{equation} \label{eq:bs_recur}
 v(h) = \mathsf a + \mathsf b h + \beta \log \mb E^Q[ e^{v(h_{t+1})} | h_t = h] =: \mb T v(h) \,,
\end{equation}  
where $\mathsf a = \alpha \beta \nu_g  + \frac{1}{2}  \alpha^2 \beta \sigma^2$ and $\mathsf b = \beta (e^{\alpha \nu_j + \frac{1}{2}\alpha^2 \sigma_j^2} - 1)$. Let $\mathsf q = 1 + c \mathsf b - \beta \varphi$.

\begin{proposition} \label{prop:bs-nonunique}
Let consumption growth $g$ evolve according to (\ref{eq:bs}) and let $\mathsf q^2 - 4 c \mathsf b > 0$. Then: recursion (\ref{eq:bs_recur}) has two solutions of the form $v_i(h) = a_i + b_i h$, $i = 1,2$, where
\[
 b_1 = \frac{\mathsf q - \sqrt{\mathsf q^2 - 4 c \mathsf b}}{2 c} \,, \quad \quad b_2 = \frac{\mathsf q + \sqrt{\mathsf q^2 - 4 c \mathsf b}}{2 c} \,,
\]
and $a_i = \frac{\mathsf a - \beta \delta \log(1 - b_i c)}{1-\beta}$, $i = 1,2.$
\end{proposition}

Note that the condition $\mathsf q^2 - 4 c \mathsf b > 0$ is satisfied for the parameterization in \cite{BS}, so uniqueness fails for that parameterization.

By contrast, when the support of $h$ is truncated to some compact interval $\mc H := [0,H]$ with $H \in (0,\infty)$, one may again verify that $\mb T$ is a contraction mapping on $B(\mc H)$. Therefore, $\mb T$ has a unique fixed point in $B(\mc H)$, irrespective of the truncation level $H$. 

The stability properties of the fixed points also differ under truncation and unboundedness in this example. Under truncation, the recursion is a (global) contraction on $B(\mc H)$ so the unique fixed point is globally attracting. In the unbounded case, suppose we restrict $\mb T$ to affine functions of the form $v(h) = a + bh $. Here the two solutions $(a_i,b_i)$, $i = 1,2$, solve the recursion $(a,b) = T(a,b)$ (see the proof of Proposition \ref{prop:bs-nonunique} for a derivation), where
\[
 T(a,b) = \left( \mathsf a + \beta a - \beta \delta \log (1 - b c) \, , \,  \mathsf b  +  \frac{ \beta \varphi b }{1-b c} \right) \,.
\]
Fixed point iteration of $T$ on an initial point $(a_0,b_0)$ converges to $(a_1, b_1)$ if $b_0 < b_2$, converges to $(a_2, b_2)$ if $b_0 = b_2$, and diverges otherwise. In the latter case, iterations diverge because the tails of $a_0 + b_0 h_{t+1} + \alpha g_{t+1}$ become increasingly heavy under repeated application of $\mb T$, eventually becoming sufficiently heavy that $\mb T v$ is no longer finite.

\section{Preliminaries}
\label{s:preliminaries}

Section \ref{s:id:gen} presents a basic existence and uniqueness result which serves as a useful starting point for organizing the discussion that follows. The key condition for uniqueness is a spectral radius condition on a sub- or supergradient of the operator. In many models with forward-looking agents---including models we study in the later sections---the subgradient is a discounted conditional expectation under a distorted law of motion. We then show in Section \ref{s:sr} that the spectral radius condition holds in these models under a ``thin tail'' condition on the change-of-measure distorting the law of motion. We shall use this result to derive more primitive conditions for recursive utilities in Sections \ref{s:robust} and \ref{s:learn}.

\subsection{A basic fixed-point result}
\label{s:id:gen}

In this section we present a basic existence and uniqueness result for an operator $\mb T$ acting on a Banach lattice $\mc E$ with partial order $\leq$\,. Our only additional requirement of $\mc E$ is that it has a \emph{monotone convergence property}: any increasing sequence $\{f_n\}_{n \geq 1} \subset \mc E$ bounded above by some $g \in \mc E$ converges to some $f \leq g$. Spaces with this property include $L^p$ spaces for $1 \leq p < \infty$ and Orlicz spaces (see Section \ref{s:orlicz}). We say that $\mb T$ is \emph{monotone} if $\mb T f \leq \mb T g$ whenever $f \leq g$. A bounded linear operator $\mb D_f$ on $\mc E$ is a  \emph{subgradient} of $\mb T$ at $f$ if 
\begin{equation} \label{e:subgrad}
 \mb T g - \mb T f \geq \mb D_f (g-f)
\end{equation}
for each $g \in \mc E$, and a \emph{supergradient} of $\mb T$ at $f$ if inequality (\ref{e:subgrad}) is reversed:
\begin{equation} \label{e:supergrad}
 \mb T g - \mb T f \leq \mb D_f (g-f)
\end{equation}
for each $g \in \mc E$. Let $\|\cdot\|$ denote the norm on $\mc E$, $\|\mb D\| := \sup\{ \| \mb D f\| : f \in \mc E, \|f\| = 1\}$ denote the norm of a linear operator $\mb D$ on $\mc E$, and $\rho(\mb D;\mc E) := \lim_{n \to \infty} \| \mb D^n \|^{1/n}$ denote the spectral radius of $\mb D$, where $\mb D^n$ denotes $\mb D$ applied $n$ times in succession.

\begin{proposition}\label{p:exun}
(i)~Existence: Let $\mb T$ be a continuous and monotone operator on $\mc E$ and let there exist $\ul v, \bar v \in \mc E$ such that either (a) $\mb T \bar v \leq \bar v$ and $\{\mb T^n \bar v\}_{n \geq 1}$ is bounded from below by $\ul v$, or (b) $\mb T \ul v \geq \ul v$ and $\{\mb T^n \ul v\}_{n \geq 1}$ is bounded from above by $\bar v$. Then: $\mb T^n \bar v$ (if (a) holds) or $\mb T^n \ul v$ (if (b) holds) converges to a fixed point $v \in \mc E$ as $n \to \infty$, where $\ul v \leq v \leq \bar v$. \\
(ii)~Uniqueness: Suppose that inequality (\ref{e:subgrad}) holds at each fixed point $v \in \mc E$, or inequality (\ref{e:supergrad}) holds at each fixed point $v \in \mc E$, and $\mb D_v$ is monotone with $\rho(\mb D_v;\mc E) < 1$ for each fixed point $v \in \mc E$. Then: $\mb T$ has at most one fixed point in $\mc E$.
\end{proposition}

When uniqueness cannot be guaranteed, ordering and stability criteria can be used to refine the set of fixed points. Let $\mc V$ denote the set of fixed points of $\mb T$. Say $v$ is the \emph{smallest} (respectively \emph{largest}) fixed point of $\mb T$ if $v \leq v'$ (resp. $v \geq v'$) holds for each $v' \in \mc V$.  Say $v$ is \emph{stable} if $\rho(\mb D_{v};\mc E) < 1$ (see, e.g., \cite{Amann1976}). 

\begin{corollary}\label{cor:nonunique-gen}
Let the conditions of Proposition \ref{p:exun}(i) hold, let $\mb T$ satisfy (\ref{e:subgrad}) (resp. (\ref{e:supergrad})) at each of its fixed points, and let $v \in \mc E$ be a fixed point of $\mb T$ with $\rho(\mb D_v;\mc E) < 1$. Then: $v$ is both the smallest (resp. largest) fixed point and the unique stable fixed point of $\mb T$ in $\mc E$.
\end{corollary}

Stability of $v$ is a useful property. In the examples we consider in Sections \ref{s:robust} and \ref{s:learn}, the subgradient is of the form $\mb D_v = \beta \tilde{\mb E}$ with $\beta \in (0,1)$, where $\tilde{\mb E}$ denotes conditional expectation under a distorted probability measure. Stability ensures that discounted expected utilities under $\tilde{ \mb E}$ are finite. Stability of $v$ also helps ensure that fixed-point iteration on a neighborhood of $v$ will converge to $v$ (see Lemma \ref{lem:cge-nbhd}).

While Proposition \ref{p:exun}(i) establishes that $\mb T^n \bar v$ (if (a) holds) or $\mb T^n \ul v$ (if (b) holds) converges to a fixed point of $\mb T$ as $n \to \infty$, it is also possible to strengthen this to a (partial) global convergence result.

\begin{corollary}\label{cor:cge}
Let the conditions of Proposition \ref{p:exun} hold, with the additional restriction that $\mb T$ satisfies (\ref{e:subgrad}) if (a) holds, or (\ref{e:supergrad}) if (b) holds, at $v$. Then: for any $w \in \mc E$ for which $w \leq \bar v$ (if (a) holds) or $w \geq \ul v$ (if (b) holds), we have $\lim_{n \to \infty} \mb T^n w = v$.
\end{corollary}

We conclude this subsection by noting results similar to Proposition \ref{p:exun} appear in the existing literature. Proposition \ref{p:exun}(i) is based on Theorem 4.1(b) of \cite{KrasPos}, which assumes the order interval $[\ul v, \bar v]$ be invariant under $\mb T$. This invariance is not necessary: all that is required is that either $\mb T \bar v \leq \bar v$ or $\mb T \ul v \geq \mb T \ul v$ and the sequence of iterates $\mb T^n \bar v$ or $\mb T^n \ul v$ is bounded from below or above, respectively.\footnote{Our requirement that $\mc E$ has the monotone convergence property is equivalent to the requirement from \cite{KrasPos} that the cone of non-negative functions is ``regular''.} Proposition \ref{p:exun}(ii) uses similar techniques to the literature on fixed points of order-convex maps (see, e.g., Chapter 5 of \cite{Amann1976}). Unlike much of this literature, Proposition \ref{p:exun}(ii) does not require additional properties such as compactness and differentiability of $\mb T$ or strict positivity of $\mb D_v$, which may be difficult to verify in practice, or that the cone of non-negative functions in $\mc E$ has non-empty interior, which is a property not shared by $L^p$ spaces with $1 \leq p < \infty$ or Orlicz classes. 
We do not view Proposition \ref{p:exun} as a contribution of this paper: we use it simply as a starting point to derive more primitive existence and uniqueness conditions in the following sections.

\subsection{Thin-tailed classes of functions}
\label{s:orlicz}

In the applications that follow, we will use a class of ``thin-tailed'' functions for the space $\mc E$. The class is naturally compatible with the structure of the recursions we study, which involve the compositions of exponentials, expectations, and logarithms. 

Let $\mu$ be a probability measure on $(\mc X, \mcr X)$. In most of the applications that follow, $\mc X$ will be the state space and $\mu$ will be the stationary distribution of the Markov state vector. Let $L^0$ denote the (equivalence class of) all measurable functions on $\mc X$. For $r \geq 1$, define
\begin{equation*}
 \begin{aligned}
 L^{\phi_r} & = \{f \in L^0 : \mb E^\mu[\exp(|f(X)/c|^r)] < \infty \mbox{ for some } c > 0\} \,, \\
 E^{\phi_r} & = \{f \in L^0 : \mb E^\mu[\exp(|f(X)/c|^r)] < \infty \mbox{ for all } c > 0\} \,,
 \end{aligned}
\end{equation*}
where $\mb E^\mu[\,\cdot\,]$ denotes expectation under $\mu$. Both $L^{\phi_r}$ and $E^{\phi_r}$ are Banach lattices when equipped with the (Luxemburg) norm
\[
 \|f\|_{\phi_r} = \inf \left\{ c > 0 : \mb E^\mu[\exp(|f(X)/c|^r)] \leq 2 \right\}
\]
and the partial order $f \geq g$ if and only if $f(x) \geq g(x)$ $\mu$-almost everywhere. The space $L^{\phi_r}$ is an \emph{(exponential) Orlicz space} and $E^{\phi_r}$ is its \emph{Orlicz heart}. We will be primarily concerned with $E^{\phi_r}$ in what follows. 

Before proceeding, we note some properties of $L^{\phi_r}$ and $E^{\phi_r}$. First, these spaces are related to $L^p(\mu)$ spaces by the embeddings $L^\infty(\mu) \hookrightarrow E^{\phi_r} \hookrightarrow L^{\phi_r} \hookrightarrow E^{\phi_s} \hookrightarrow L^{\phi_s}\hookrightarrow L^p(\mu)$ for $1 \leq s < r < \infty$, with $\|f\|_p \leq p!(\log 2)^{1/r-1} \|f\|_{\phi_r}$ for each $1 \leq p < \infty$ where $\|\cdot\|_p$ denotes the $L^p(\mu)$ norm, and $\|f\|_{\phi_{s}} \leq  (\log 2)^{1/r-1/s}\|f\|_{\phi_{r}}$ \cite[p. 95]{vdVW}. In addition, Lemma \ref{lem:mcp} shows that $E^{\phi_r}$ has the monotone convergence property. We refer the reader to \cite{KrasRuti} for further details on Orlicz classes.

\subsection{Verifying the spectral radius condition}
\label{s:sr}

In many models featuring forward-looking agents such as those we study in Sections \ref{s:robust} and \ref{s:learn}, the subgradient is a discounted conditional expectation operator under a distorted probability measure. That is, there is a wedge between the probability measure describing the evolution of state variables and the probability measure under which the expectation is taken. In this section we show how to verify the key spectral radius condition from Proposition \ref{p:exun} under a thin-tail condition on the change of measure.

When there is no such wedge (e.g., time-separable preferences and rational expectations), the spectral radius condition is easily seen to hold. Let $X = \{X_t\}_{t \geq 0}$ be a time-homogeneous Markov process with transition kernel $Q$ and stationary distribution $\mu$. Suppose $\mb D_v = \beta \mb E^Q$, where $\mb E^Q$ denotes conditional expectation under $Q$. Then for any $c > 0$ and $f \in E^{\phi_r}$, 
\begin{align*}
 \mb E^\mu[\exp(| \mb D_v f(X_t)/ (\beta c)|^r)] 
 & = \mb E^\mu[\exp(|\mb E^Q[f(X_{t+1})|X_t]/c|^r)] \\
 & \leq \mb E^\mu[\mb E^Q[\exp(|f(X_{t+1})|^r/c)|X_t]] \\
 & = \mb E^\mu[\exp(|f(X_t)|^r/c)] \,,
\end{align*}
by Jensen's inequality and the fact that $\mu$ is the stationary distribution associated with $X$. Taking $f$ to be almost-everywhere constant, we see that the operator $\mb D_v$ has norm $\| \mb D_v \|_{\phi_r} = \beta$ on $E^{\phi_r}$ and $\rho(\mb D_v ; E^{\phi_r}) = \beta$. A similar argument applies for $L^p(\mu)$ spaces.

This argument breaks down in the settings we study, in which $\mb D_v = \beta \tilde{\mb E}$, where $\tilde{\mb E}$ denotes conditional expectation under a distribution different from $Q$. Suppose
\begin{equation} \label{e:rn-sufficient}
  \tilde{\mb E} f(x) =  \mb E^Q[ m(X_t,X_{t+1}) f(X_{t+1}) |X_t = x] \,,
\end{equation}
where $m$ is the (conditional) change-of-measure transforming $\mb E^Q$ into $\tilde{\mb E}$. We shall verify the spectral radius condition under a thin-tail condition on $m$. For the intuition behind the result, note that applying $\mb D_v$ involves multiplying by $m$, taking conditional expectations under $Q$, and discounting. Therefore, provided the higher moments of $m$ don't diverge too quickly, repeatedly applying $\mb D_v$ to thin-tailed functions ensures that the effect of discounting eventually dominates and the spectral radius condition holds.

To formalize this reasoning, let $\log m \vee 0$ denote the pointwise maximum of $\log m$ and $0$ and let $\mu \otimes Q$ denote the joint (stationary) distribution of $(X_t,X_{t+1})$.

\begin{lemma}\label{lem:com}
Let $\mb D = \beta \tilde{\mb E}$ where $\beta \in (0,1)$ and $\tilde{\mb E}$ is of the form (\ref{e:rn-sufficient}) with
\begin{equation}\label{e:m-thin}
 \mb E^{\mu \otimes Q}\left[ \exp(|\!\log m(X_t,X_{t+1}) \vee 0|^r/c) \right] < \infty
\end{equation}
for some $c > 0$ and $r > 1$. Then: $\mb D$ is a bounded linear operator on $E^{\phi_s}$ with $\rho(\mb D;E^{\phi_s}) < 1$ for each $s \geq 1$.
\end{lemma}

\begin{remark}
Lemma \ref{lem:com} does not require stationarity (or any other property) of $X$ under the law of motion corresponding to $\tilde{\mb E}$.
\end{remark}

\begin{remark}
Lemma \ref{lem:com} establishes the spectral radius condition for all $\beta \in (0,1)$. When the change of measure $m$ defining $\tilde{\mb E}$ has thin tails, any amount of discounting is sufficient to overwhelm the effect of the change of measure under repeated application of $\mathbb D = \beta \tilde{\mb E}$.
\end{remark}

\section{Application 1: Robust (and related) preferences}
\label{s:robust}

\subsection{Setting}

Consider an infinite-horizon environment in which the continuation value $V_t$ of a stream of per-period utilities $\{U_t\}_{t \geq 0}$ from date $t$ forwards is defined recursively by
\begin{equation}\label{e:recur}
 V_t = U_t - \beta \theta \log \mb E\left[ \left. e^{-\theta^{-1}V_{t+1}} \right| \mc F_t \right] \,,
\end{equation}
where $\mc F_t$ is the date-$t$ information set, $\beta \in (0,1)$ is a time preference parameter, and $\theta > 0$. Recursion (\ref{e:recur}) arises in a number of settings. It is the risk-sensitive recursion of \cite{HS1995}, where $\theta$ is interpreted as a risk-sensitivity parameter. The recursion also arises under ``robust'' preferences which express an aversion to model uncertainty, namely multiplier preferences \citep{HS2001} and constraint preferences \citep{HSTW}, in which $\theta$ encodes the agent's aversion to model uncertainty. 
Finally, recursion (\ref{e:recur}) is equivalent to \cite{EZ} preferences with unit IES, in which case $\theta$ is a transformation of the risk aversion parameter.\footnote{Specifically, $\theta = 1/(\gamma - 1)$ where $\gamma$ is the coefficient of relative risk aversion. See, e.g., Section III in \cite{HHL} for a derivation of recursion (\ref{e:recur}) from the Epstein--Zin recursion with unit IES.}

We follow much of the literature and consider environments characterized by a stationary Markov state process $X = \{X_t : t \geq 0\}$ supported on a state space $\mathcal X \subseteq \mb R^d$. The set $\mc F_t$ will denote the information set generated by the realization of $X$ up to date $t$. Let $Q$ denote the Markov transition kernel and $\mb E^Q$ denote conditional expectation with respect to $Q$.  In such environments it follows for certain commonly used specifications of $U_t$ that there exists $v : \mc X \to \mb R$ and $u : \mc X \times \mc X \to \mb R$ and  such that
\begin{align*}
 v(X_t) & = -\frac{1}{\theta} \left( V_t - \frac{1}{1-\beta} U_t \right)\,, &
 u(X_t,X_{t+1}) & =  U_{t+1} - U_t \,.
\end{align*} 
For instance, this is true when $U_t = \log (C_t)$ and consumption growth $\log(C_{t+1}/C_t)$ is a function of $(X_t,X_{t+1})$.\footnote{Our results trivially extend to allow $\log(C_{t+1}/C_t) = g(X_t,X_{t+1},Y_{t+1})$ where the conditional distribution of $(X_{t+1},Y_{t+1})$ given $(X_t,Y_t)$ depends only on $X_t$ by redefining the state as $(X_t,Y_t)$.} Under these conditions, the recursion may be rewritten in terms of the scaled continuation value function $v$:
\begin{equation} \label{e:recur-2}
 v(x) = \beta \log \mb E^{Q} \left[ \left. e^{v(X_{t+1}) + \alpha u(X_t,X_{t+1})} \right| X_t = x \right] \,,
\end{equation}
where $\alpha = -(\theta(1-\beta))^{-1}$. Recursion (\ref{e:recur-2}) may be expressed as $v = \mb T v$, where 
\[
 \mb T f(x) = \beta \log \mb E^{Q} \left[ \left. e^{f(X_{t+1}) + \alpha u(X_t,X_{t+1})} \right| X_t = x \right] \,.
\]

\subsection{Existing results}

\cite{HS2012} and \cite{npsdfd} studied this recursion in the context of Epstein--Zin preferences with unit IES and unbounded $\mc X$. \cite{HS2012} derived sufficient conditions for existence of a fixed point but not uniqueness. Their conditions restrict moments of a Perron--Frobenius eigenfunction of an operator and require convergence of a sequence of iterates of a related recursion.  \cite{npsdfd} established  uniqueness on a neighborhood for the same recursion under a spectral radius condition but did not establish existence or global uniqueness.

\subsection{New results}

Here we establish existence and uniqueness under a primitive thin-tail condition on the growth in per-period utility. Formally, we require that for some $r \geq 1$, 
\begin{equation} \label{e:u}
 \mb E^{\mu \otimes Q}\left[ \exp(|u(X_t,X_{t+1})|^r/c) \right] < \infty  \quad \mbox{ for all $c > 0$.}
\end{equation}
We verify this condition below in several examples. Note, however, that both examples in Section \ref{sec:problems} violate this condition.

We shall establish existence and uniqueness by applying Proposition \ref{p:exun}. The operator $\mb T$ is continuous, monotone, and convex under condition (\ref{e:u}); see Lemma \ref{lem:T-prop}. The proof of existence constructs an upper value $\bar v$ and shows the sequence of iterates $\{\mb T^n \bar v\}_{n \geq 1}$ is bounded from below. These steps use nothing more than repeated application of H\"older's inequality and Jensen's inequality. 
For uniqueness, by Jensen's inequality the operator $\mb T$ satisfies inequality (\ref{e:subgrad}) with subgradient
\[
 \mb D_v f(x) = \beta \mb E_v f(x)\,,
\]
where $\mb E_v$ is a distorted conditional expectation:
\begin{align}
 \mb E_v f(x) & = \mb E^Q[m_v(X_t,X_{t+1})f(X_{t+1})|X_t = x] \,, \label{eq:mf-robust} \\
  m_v(X_t,X_{t+1}) & = \frac{e^{v(X_{t+1}) + \alpha u(X_t,X_{t+1})}}{\mb E^Q[e^{v(X_{t+1}) + \alpha u(X_t,X_{t+1})}|X_t]} \,.
\end{align}
For robust preferences, $\mb E_v$ may be interpreted as expectation under the agent's ``worst-case'' model. The spectral radius condition is verified by applying Lemma \ref{lem:com}; see Lemma \ref{lem:D-bdd}.

\begin{theorem}\label{t-id-W}
Let condition (\ref{e:u}) hold. Then: $\mb T$ has a fixed point $v \in E^{\phi_r}$. Moreover, if $r > 1$ then: (i) $v$ is the unique fixed point of $\mb T$ in $ E^{\phi_s}$ for each $s \in (1,r]$, and (ii) $v$ is both the smallest fixed point and the unique stable fixed point of $\mb T$ in $E^{\phi_1}$.
\end{theorem}

\medskip

\paragraph{Example 1: Linear-Gaussian environments.} Condition (\ref{e:u}) holds for all $r \in [1,2)$ when $u(X_t,X_{t+1}) = \lambda_0'X_t + \lambda_1'X_{t+1}$ and its stationary distribution is Gaussian.   

This specification arises, for instance, with $U_t = \log(C_te^{\lambda'X_t})$ where $\log(C_{t+1}/C_t)$ is a function of $(X_t,X_{t+1})$ and the process $X$ is a stationary Gaussian VAR(1):
\[
 X_{t+1} =  \nu + A X_t + u_{t+1}\,,  \quad u_{t+1} \sim N(0,\Sigma) \,,
\]
with all eigenvalues of $A$ inside the unit circle. This setting was considered in \cite{HHL}, \cite{BHS}, and several other works. It is known that $\mb T$ has a fixed point of the form $v(x) = a + b 'x$ where $b = \alpha \beta (I - \beta A')^{-1}(\lambda_0 + A'\lambda_1) $ and 
\[
 a = \frac{\beta}{1-\beta} \Big( (\alpha \lambda_1 + b)'\nu + \frac{1}{2}(\alpha \lambda_1 + b)' \Sigma(\alpha \lambda_1 + b) \Big)  \,.
\]
Theorem \ref{t-id-W} shows that $v(x)=a+b'x$ is the unique fixed point in $E^{\phi_s}$ for all $s \in (1,2)$, and the smallest fixed point and unique stable fixed point in $E^{\phi_1}$. \hfill $\square$

\medskip

\paragraph{Example 2: Fat tails and rare disasters.} Consider the model from Section \ref{sec:bs}. 
Here with $X_t = (g_t,h_t)$ we have $u(X_t,X_{t+1}) = g_{t+1}$. 
By iterated expectations we may deduce
\[
 \mb E^{\mu \otimes Q}\left[ e^{c u(X_t,X_{t+1})}\right] = e^{c \nu_g + \frac{c^2 \sigma^2}{2}} \mb E^{\mu} \left[ \exp \left( h_t \left( {\textstyle \exp\left\{c \nu_j + \frac{c^2 \sigma_j^2}{2}\right\} - 1 } \right) \right) \right] \,.
\]
Condition (\ref{e:u}) is violated for this model: the expectation on the right-hand side is only finite if $c$ is in a neighborhood of zero because the stationary distribution of $h_t$ is a Gamma distribution. Note that uniqueness can fail for this model, as illustrated in Section \ref{sec:bs}.

One could modify this specification so that $w_{z,t+1}|j_{t+1} \sim N(\nu_j j_{t+1}^{\varsigma}, \sigma_j^2 )$ for some $\varsigma \in [\frac{1}{2},1)$. Given the low frequency of jumps, this modification is likely difficult to distinguish empirically from the original specification. Under this modification, condition (\ref{e:u}) holds for each $r \in [1, 1/\varsigma)$. Therefore, there is a unique fixed point $v \in E^{\phi_s}$ for all $s \in (1, 1/\varsigma)$, and $v$ is both the smallest fixed point and the unique stable fixed point in $E^{\phi_1}$. \hfill $\square$

\medskip

\paragraph{Example 3: Regime-switching.} Consider the same setup from Example 1 but suppose now that the parameters of the VAR are state-dependent (see, e.g., \cite{Hamilton1989}, \cite{CLM1990,CLM2000}, \cite{HS2010}, and \cite{AngTimmermann}):
\[
 X_{t+1} =  \nu_{s_t} + A_{s_t} X_t + u_{t+1}\,,  \quad u_{t+1} \sim N(0,\Sigma_{s_t}) \,,
\]
where $s_t$ is stationary, exogenous Markov state taking values in $\{1, \ldots,N\}$, and all eigenvalues of $A_s$ are inside the unit circle for each $s = 1,\ldots,N$. The full state vector is now $(X_t,s_t)$, which is Markovian and stationary. The stationary distribution of growth in per-period utilities $u(X_t,X_{t+1})$ is sub-Gaussian (see, e.g., \citeauthor{Vershynin}, \citeyear{Vershynin}, Section 2.5), and so condition (\ref{e:u}) holds for all $r \in [1,2)$. It follows by Theorem \ref{t-id-W}  there is a unique fixed point $v \in E^{\phi_s}$ for all $s \in (1,2)$ (with $E^{\phi_s}$ defined with respect to the stationary distribution of $(X_t,s_t$)), and $v$ is both the smallest fixed point and the unique stable fixed point in $E^{\phi_1}$. \hfill $\square$

\medskip

\paragraph{Example 4: Stochastic volatility.}
Consider the environment from section I.B of \cite{BansalYaron2004} in which consumption growth $g_{t+1}:=\log(C_{t+1}/C_t)$ is modeled as
\begin{align*}
 g_{t+1} & = \bar g + x_t + \sigma_t \eta^g_{t+1} \,, \\
 x_{t+1} & = \rho_x x_t + \varphi_x \sigma_t \eta^x_{t+1} \,, \\
 \sigma_{t+1}^2 & = \bar \sigma^2 + \rho_{\sigma} (\sigma_t^2 - \bar \sigma^2) + \varphi_{\sigma} \eta^{\sigma}_{t+1} \,,
\end{align*}
where $\eta^g_t$, $\eta^x_t$, and $\eta^\sigma_t$ are all i.i.d. $N(0,1)$. We alter this model slightly in two respects. First, to focus on the implications of stochastic volatility and simplify exposition we set $\rho_x = 0$ though this is not essential to our analysis. Second, to deal with the complications arising when $\sigma_{t+1}^2 < 0$ we take absolute values. This leads to the consumption growth process
\begin{align*}
 g_{t+1} & = \bar g + \sqrt{|s_t|}  \eta_{t+1}^g \,, \\
 s_{t+1} & = \bar s + \rho_s(s_t - \bar s) + \varphi_s \eta_{t+1}^s \,,
\end{align*}
where $\eta^g_t$ and $\eta^s_t$ are i.i.d. $N(0,1)$. Defining $X_t = (g_t, s_t)$, we see that $u(X_t,X_{t+1}) = g_{t+1}$ when per-period utility is logarithmic in consumption. To verify condition (\ref{e:u}), first note that 
\begin{equation} \label{eq:sv}
 \mathbb E^{\mu \otimes Q}[\exp(|(g_{t+1} - \bar g)/c|^r) ] = \mb E^\mu[ \mb E[\exp(|\sqrt{|s_t|}  \eta_{t+1}^g/c|^r)|s_t]] \,,
\end{equation} 
where the inner expectation is taken with respect to $\eta_{t+1}^g \sim N(0,1)$. 
The inner expectation is equivalent to $\mathbb E[\exp(Y^r/a^r)]$ where $Y = |Z|$ with $Z \sim N(0,1)$ and $a = c/\sqrt{|s_t|} > 0$. In Appendix \ref{s:proofs} we derive a crude bound on this expectation (see Lemma \ref{lem:hn}) from which we may deduce that for $r \in [1,2)$,
\begin{multline*}
  \mb E\left[ \left. \exp\left(\left|\frac{\sqrt{|s_t|}\eta_{t+1}^g}{c}\right|^r \right) \right|s_t \right] \\
  \leq \frac{\sqrt 2}{\sqrt \pi} \left(  \left(\frac{2\sqrt{|s_t|}^r}{c^r}\right)^\frac{1}{2-r} \exp \left( \frac{(2 |s_t|)^\frac{r}{2-r} }{c^\frac{2r}{2-r}} \right) + \left(\frac{4\sqrt{|s_t|}^r}{c^r}\right)^\frac{1}{2-r} + \sqrt \pi \right) .
\end{multline*}
As the stationary distribution of $s_t$ is Gaussian, the exponent $\frac{r}{2-r}$ of the $|s_t|$ term appearing in the right-hand side exponential must be less than $2$ (equivalently, $r \in [1,4/3)$) so that that the expectation (\ref{eq:sv}) is finite for all $c > 0$. It follows that (\ref{e:u}) holds for all $r \in [1,4/3)$. Therefore, there is a unique fixed point in $v \in E^{\phi_s}$ for all $s \in (1,4/3)$, and $v$ is both the smallest fixed point and the unique stable fixed point in $E^{\phi_1}$. \hfill $\square$

\subsection{Convergence of compact approximations}

While there are many different ways to construct versions of $\mb T$ over truncated state spaces, a natural approach is to simply restrict $\mc X$ to a large but compact set $\mc C$ and rescale the transition density of $X$ accordingly. We close this section by showing that this construction yields an operator $\mb T_{\mc C}$ whose fixed point $v_{\mc C}$ approaches the unique stable fixed point $v$ of $\mb T$ from below as $\mc C$ becomes large. In view of Theorem \ref{t-id-W}(ii), this result implies that $v_{\mc C}$ will \emph{not} converge to any unstable fixed point of $\mb T$ (if unstable fixed points of $\mb T$ do indeed exist).

Let $\mc C \subset \mc X$ be a compact set and define
\[
 \mb T_{\mc C} f(x) = \beta \log \mb E^Q \left[ \left. e^{f(X_{t+1}) + \alpha u(X_t,X_{t+1})} \frac{\ind \{X_{t+1} \in \mc C \}}{Q(\mc C|x)} \right| X_t = x \right] \,, \quad x \in \mc C \,,
\]
where $Q(\mc C|x)$ is the conditional probability (under the un-truncated law of motion $Q$) that $X_{t+1} \in \mc C$ given $X_t = x$ and $\ind\{x \in \mc C\} = 1$ if $x \in \mc C$ and $0$ otherwise. The operator $\mb T_{\mc C}$ is defined by simply truncating the support of $X$ to $\mc C$ and rescaling the transition distribution $Q$ accordingly. Let $B(\mc C)$ denote the space of bounded functions on $\mc C$ under the sup-norm. 

\begin{proposition}\label{prop:cpt-approx}
Let $\sup_{x \in \mc C} |\log \mb E^Q [ e^{\alpha u(X_t,X_{t+1})} {\ind \{X_{t+1} \in \mc C \}}/{Q(\mc C|x)} | X_t = x ]| < \infty$. Then: $\mb T_{\mc C}$ has a unique fixed point $v_{\mc C} \in B(\mc C)$. Moreover, if $\inf_{x \in \mc C}Q(\mc C|x) > 0$ then for any fixed point $v$ of $\mb T$, 
\[
 \inf_{x \in \mc C} \left( v(x) - v_{\mc C}(x) \right) \geq \frac{\beta}{1-\beta} \left( \inf_{x \in \mc C} \log Q(\mc C|x) \right) .
\]
\end{proposition}

As $\epsilon_{\mc C} := -\frac{\beta}{1-\beta} \left( \inf_{x \in \mc C} \log Q(\mc C|x) \right) > 0$, Proposition \ref{prop:cpt-approx} implies $v_{\mc C}(x) \leq v(x) + \epsilon_{\mc C}$ holds for all $x \in \mc C$. If $\mb T$ has a second (unstable) fixed point $v' \geq v$, then for any subset of $\mc C$ upon which $v'$ and $v$ differ by more than $\epsilon_{\mc C}$, we have $v_{\mc C}(x) \leq v(x) + \epsilon_{\mc C} < v'(x)$. As such, $v_{\mc C}$ cannot converge to $v'$ as $\mc C$ becomes large (i.e., as $\epsilon_{\mc C} \to 0$).

\section{Application 2: Learning and ambiguity}
\label{s:learn}

We now extend the setting from Section \ref{s:robust} to models in which the agent learns about a hidden state, e.g. a regime, stochastic volatility, growth process, or time-varying parameter. This setting is relevant for several types of preferences, including: (i) the extension of multiplier preferences by \cite{HS2007,HS2010} to include concerns about misspecification of beliefs about the hidden state, (ii) generalized recursive smooth ambiguity preferences of \cite{JuMiao} with unit IES, (iii) special cases of recursive smooth ambiguity preferences studied by \cite{KMM2009}, and (iv) \cite{EZ} recursive preferences with unit IES and learning as used, for example, by \cite{CLL}. 

\subsection{Setting}

We again consider environments characterized by a Markov state process $X = \{X_t\}_{t \geq 0}$ with transition kernel $Q$. Partition the state as $X_t = (\varphi_t,\xi_t)$ where the agent observes $\varphi_t$ but does not observe $\xi_t$. Let $\mc O_t = \sigma(\varphi_t,\varphi_{t-1},\ldots,\varphi_0)$ denote the history of the observed state to date $t$. Beliefs about $\xi_t$ are summarized by a posterior distribution $\Pi_t$ conditional on $\mc O_t$. We consider environments in which the continuation value $V_t$ of a stream of per-period utilities $\{U_t\}_{t \geq 0}$ from date $t$ forward is defined recursively as
\begin{equation} \label{e-V-recur-learn}
 V_t = U_t - \beta \theta \log \mb E^{\Pi_t}\! \left[ \left. \mb E^Q \left[ \left.  e^{-\vartheta^{-1} V_{t+1}} \right| \mc O_t,\xi_t \right]^\frac{\vartheta}{\theta}  \right| \mc O_t \right] \,,
\end{equation}
for $\beta \in (0,1)$. 
This recursion is from \cite{HS2007,HS2010}, who introduce an extension of multiplier preferences to accommodate concerns about misspecification of the model ($Q$) and beliefs about the hidden state ($\Pi_t$), where $\vartheta>0$ and $\theta>0$ encode concerns about misspecification of $Q$ and $\Pi_t$, respectively. When $U_t = \log C_t$, recursion (\ref{e-V-recur-learn}) also arises under generalized recursive smooth ambiguity preferences of \cite{JuMiao} with unit IES, where $\theta$ and $\vartheta$ are one-to-one transformations of their ambiguity aversion and risk aversion parameters, respectively. When $\vartheta = \theta$, recursion (\ref{e-V-recur-learn}) reduces to
\[
 V_t = U_t - \beta \vartheta \log \mb E^{\Pi_t}\! \left[ \left. \mb E^Q \left[ \left.  e^{-\vartheta^{-1} V_{t+1}} \right| \mc O_t,\xi_t \right]  \right| \mc O_t \right] \,.
\]
With $U_t = \log C_t$, this recursion corresponds to Epstein--Zin recursive preferences with unit IES and learning about the hidden state. 
In the limit as $\vartheta \to \infty$ (thus, the agent is confident in $Q$ but has doubts about the hidden state) recursion (\ref{e-V-recur-learn}) becomes
\begin{equation} \label{e-V-recur-learn-limit}
 V_t = U_t - \beta \theta \log \mb E^{\Pi_t}\! \left[ \left. e^{-\theta^{-1} \mb E^Q \left[ \left.  V_{t+1}  \right| \mc O_t,\xi_t \right] }  \right| \mc O_t \right] \,.
\end{equation}
This recursion is obtained under recursive smooth ambiguity preferences of \cite{KMM2009}, when their $\phi$ function is $\phi(x)=\exp(-\theta^{-1}x)$. 

We impose several (standard) conditions to make the problem tractable. First, the state is assumed to have a conventional hidden Markov structure, in which
\[
 Q(X_{t+1}|X_t) = Q_\varphi(\varphi_{t+1}|\xi_t)Q_\xi(\xi_{t+1}|\xi_t)\,.
\]
This nests models with regime-switching studied by \cite{JuMiao} as well as models with learning about a hidden growth term as in \cite{HS2007,HS2010}, \cite{CLL} and \cite{CMST}. 
Our analysis extends to allow $\varphi_t$ to influence $\varphi_{t+1}$, but we maintain this simpler presentation for convenience. 

Second, we assume $\Pi_t$ is summarized by a finite-dimensional sufficient statistic $\hat \xi_t$:
\[
 \Pi_t(\xi_t) = \Pi_\xi(\xi_t|\hat \xi_t) 
\]
for some conditional distribution $\Pi_\xi$, where $\hat \xi$ is updated according to a time-invariant rule: 
\[
 \hat \xi_{t+1} = \Xi(\hat \xi_t,\varphi_{t+1})\,.
\]
These conditions are satisfied under Bayesian updating when the state $\xi_t$ takes finitely many values (e.g. a hidden regime) and when $X_t$ evolves as a Gaussian state-space model; see below.
The rule for $\hat \xi_t$ could also represent belief updating in a boundedly-rational way. Let $\hat X_t = (\varphi_t,\hat \xi_t)$ and let $\mc X_{\hat X}$, $\mc X_{\hat \xi}$, and $\mc X_\varphi$ denote the support of $\hat X_t$, $\hat \xi_t$, and $\varphi_t$. 
 
We assume learning is in a ``steady state'', i.e., $\{(\xi_t,\hat X_t) \}_{t \geq 0}$ is stationary. In linear-Gaussian environments, learning corresponds to the Kalman filter. If the filter is not initialized in its steady-state then this process will typically be non-stationary. The stationary problem studied here is a boundary problem representing convergence of the filter to its steady state. Solutions can be obtained by backwards iteration from the steady-state boundary solution.\footnote{A similar approach is taken by \cite{CDJL} in models featuring Epstein--Zin preferences and learning about parameters of the data-generating process.} Uniqueness of the limiting steady state recursion is necessary for uniqueness of the sequence of backward iterates. 

Finally, we require that there exists $v : \mc X_{\hat \xi} \to \mb R$ and $u : \mc X_\varphi \to \mb R$ such that
\begin{align*}
 v(\hat \xi_t) & = -\frac{1}{\theta} \left( V_t - \frac{1}{1-\beta} U_t \right)\,, &
 u(\varphi_{t+1}) & =  U_{t+1} - U_t \,.
\end{align*} 
Before proceeding, we give two examples of environments in which the preceding conditions hold. In both examples, $U_t = \log (C_t)$ and $\log(C_{t+1}/C_t)$ is a function of $\varphi_{t+1}$.

\medskip 

\paragraph{Example 1: Regime switching.}
Suppose that $\xi_t \in \{1,\ldots,N\}$ denotes a hidden Markov state with transition matrix $ \Lambda$. Let the conditional distribution of $\varphi_{t+1}$ given $\xi_t = \xi$ have density $q(\cdot|\xi)$. The posterior $\Pi_t$ is identified with a vector $\hat \xi_t$ of regime probabilities given $\mc O_t$. Beliefs $\hat \xi_t$ are updated as
\[
 \hat \xi_{t+1} =   \Lambda \frac{ q(\varphi_{t+1}) \odot \hat \xi_t}{ 1' ( q(\varphi_{t+1}) \odot \hat \xi_t)} \,,
\]
where $ q(\varphi_{t+1})$ is the $N$-vector whose entries are $q(\varphi_{t+1}|\xi)$ for $\xi \in \{1,\ldots,N\}$, $\odot$ denotes element-wise product, and $ 1$ is a $N$-vector of ones (see, e.g., \citeauthor{Hamilton1994}, \citeyear{Hamilton1994}, Section 4.2). 

For example, \cite{JuMiao} study an economy in which consumption and dividend growth is jointly dependent on a hidden regime $\xi_t$:
\[
 \log(C_{t+1}/C_t)  = \kappa_{\xi_t} +  u_{t+1}^C \,, \quad 
 \log(D_{t+1}/D_t)  = \zeta \log(C_{t+1}/C_t) + g_d + u_{t+1}^D \,,
\]
where $u_t^C$ and $u_t^D$ are i.i.d. $N(0,\sigma_C^2)$ and $N(0,\sigma_D^2)$. The observable state is $\varphi_t = \log(C_t/C_{t-1})$. The stationary distribution of $u(\varphi_{t+1})$ is a finite mixture of Gaussians. Our results also allow the \emph{volatility} of consumption and dividend growth to be state-dependent. \hfill $\square$

\paragraph{Example 2: Gaussian state-space models.}
Suppose $X$ evolves under $Q$ according to:
\[
 \varphi_{t+1}  =  A \xi_t + u_{t+1}^\varphi \,, \quad 
 \xi_{t+1} = B \xi_t + u_{t+1}^\xi \,,
\]
where $u_t^\varphi$ and $u_t^\xi$ are i.i.d. $N(0,\Sigma_u)$ and $N(0,\Sigma_w)$ and all eigenvalues of $B$ are inside the unit circle. This is the setting studied in  \cite{HS2007,HS2010},  \cite{CLL}, \cite{CMST}, and several other works.  If $\xi_0\sim N(\hat \mu_{0},\hat \Sigma_{0})$ under $\Pi_0$ then $\xi_t \sim N(\hat \mu_{t},\hat \Sigma_{t})$ under $\Pi_t$. The matrix $\hat \Sigma_{t}$ will converge to a fixed matrix $\bar \Sigma$ as $t \to \infty$. In this steady state, the sufficient statistic for $\Pi_t$ is $\hat \xi_t = \hat \mu_{t}$ which is updated using 
\[
 \hat \xi_{t+1} = B \hat \xi_{t} + B \bar\Sigma A'(A\bar\Sigma A' + \Sigma_u)^{-1}(\varphi_{t+1}  -  A \hat \xi_{t})\,.
\]
The stationary distribution of $u(\varphi_t)$ is Gaussian. \hfill $\square$

\subsection{Existing results}

The only related existence and uniqueness result we are aware of in any of these setting is that of \cite{KMM2009} for recursive smooth ambiguity preferences (recursion (\ref{e-V-recur-learn-limit})). Their result applies to bounded functions and requires bounded per-period utilities.

\subsection{New results}
 
Recursion (\ref{e-V-recur-learn}) may be reformulated as the fixed-point equation $v = \mb T v$ where
\[
 \mb T f(\hat \xi_t) = \beta \log \mb E^{\Pi_\xi}\! \left[ \left. \mb E^{Q_\varphi} \left[ \left.  e^{\frac{\theta}{\vartheta} f(\Xi(\hat \xi_t,\varphi_{t+1})) + \alpha u(\varphi_{t+1})} \right| \xi_t,\hat \xi_t \right]^\frac{\vartheta}{\theta}  \right| \hat \xi_t \right] .
\]
Recursion (\ref{e-V-recur-learn-limit}) in the limiting case with $\vartheta = +\infty$ may be reformulated as $v = \mb T v$ where
\[
 \mb T f(\hat \xi_t) = \beta \log \mb E^{\Pi_\xi}\! \left[ \left.   e^{ \mb E^{Q_\varphi} \left[ \left.  f(\Xi(\hat \xi_t,\varphi_{t+1})) + \alpha u(\varphi_{t+1})  \right| \xi_t,\hat \xi_t \right] }   \right| \hat \xi_t \right] \,.
\]
The existence and uniqueness results presented below apply to either case, though the proofs are presented only for the more involved setting in which $\vartheta < \infty$.

Let $E^{\phi_r}_{\hat X}$ be defined relative to the stationary distribution $\mu$ of $\hat X_t = (\varphi_t',\hat \xi_t')'$. Similarly, let $E^{\phi_r}_\varphi \subset E^{\phi_r}_{\hat X}$ and $E^{\phi_r}_{\hat \xi} \subset E^{\phi_r}_{\hat X}$ denote functions in $E^{\phi_r}_{\hat X}$ depending only on $\varphi$ or $\hat \xi$, respectively. The key regularity condition is again that the stationary distribution of utility growth has thin tails:
\begin{equation}\label{e:u-learn}
 u \in E^{\phi_r}_\varphi 
\end{equation}
for some $r \geq 1$.
Note that this condition depends only on the marginal distribution of the observed state and is therefore easy to verify.

We establish existence and uniqueness of fixed points of $\mb T$ by applying Proposition \ref{p:exun}. 
Further details on the form of the subgradient and verification of Lemma \ref{lem:com} are deferred to Appendix \ref{s:proofs-learn}.

\begin{theorem}\label{t-id-W-learn}
Let condition (\ref{e:u-learn}) hold. Then: $\mb T$ has a fixed point $v \in E^{\phi_r}_{\hat \xi}$. Moreover, if $r > 1$, then: (i) $v$ is the unique fixed point of $\mb T$ in $ E^{\phi_s}_{\hat \xi}$ for all $s \in (1, r]$, and (ii)  $v$ is both the smallest fixed point and the unique stable fixed point of $\mb T$ in $E^{\phi_1}_{\hat \xi}$.
\end{theorem}

\medskip

\paragraph{Example 1: Regime switching (continued).}
In the example of \cite{JuMiao}, the stationary distribution of $u(\varphi_{t+1})$ is a finite mixture of Gaussians, so (\ref{e:u-learn}) holds for all $r \in [1, 2)$, including when the volatility of consumption and dividend growth is state-dependent. Therefore, there is a unique fixed point in $v \in E^{\phi_s}_{\hat \xi}$ for all $s \in (1,2)$, and $v$ is both the smallest fixed point and the unique stable fixed point in $E^{\phi_1}_{\hat \xi}$. \hfill $\square$

\medskip

\paragraph{Example 2: Gaussian state-space models (continued).}
Here the stationary distribution of $u(\varphi_{t+1})$ is Gaussian, so (\ref{e:u-learn}) holds for all $r \in [1, 2)$. Therefore, there is a unique fixed point in $v \in E^{\phi_s}_{\hat \xi}$ for all $s \in (1,2)$, and $v$ is both the smallest fixed point and the unique stable fixed point in $E^{\phi_1}_{\hat \xi}$. \hfill $\square$

\medskip

It is straightforward (albeit more cumbersome notationally) to extend the preceding analysis to allow for $u$ to depend on $(\varphi_t,\varphi_{t+1})$ and to allow the law of motion to be of the more general form
\[
 Q(X_{t+1}|X_t) = Q_\varphi(\varphi_{t+1}|\xi_t,\varphi_t)Q_\xi(\xi_{t+1}|\xi_t)\,.
\]
In this case, however, the effective state vector will be $\hat X_t$ rather than $\hat \xi_t$.

\section{Application 3: Epstein--Zin preferences}
\label{s:ez}

In this section we study \cite{EZ} recursive utility with IES $\neq 1$. Existence and uniqueness when state variables have non-compact support is of particular importance as many prominent models, such as those in the long-run risks literature, have non-compact state space.  There are currently no uniqueness results for the recursion we study with non-compact state space. This is a complicated issue and it is beyond the scope of the paper to provide a comprehensive treatment. Rather, we show how our approach may be used to derive primitive existence conditions in empirically relevant settings.

\subsection{Setting}

The continuation value $V_t$ of the agent's consumption plan from time $t$ forward solves
\[
 V_t = \left\{ (1-\beta) (C_t)^{1-\rho} + \beta \mb E[(V_{t+1})^{1-\gamma}| \mc F_t]^\frac{1-\rho}{1-\gamma} \right\}^\frac{1}{1-\rho} \,,
\]
where $C_t$ is date-$t$ consumption, $\mathcal F_t$ is date-$t$ information, $\gamma \in (0,1) \cup (1,\infty)$ is the coefficient of relative risk aversion, and $1/\rho>0$ is the elasticity of intertemporal substitution. 

We consider the $\rho \neq 1$ case in this section as the $\rho = 1$ case is studied in Section \ref{s:robust}. 
We again consider environments characterized by a stationary Markov process $X = \{X_t : t \geq 0\}$ with state space $\mathcal X \subseteq \mb R^d$. Let $Q$ denote the Markov transition kernel and $\mb E^Q$ denote conditional expectation under $Q$. Also let  $\log(C_{t+1}/C_t) = g(X_t,X_{t+1})$ for some function $g$.\footnote{Our results trivially extend to allow $\log(C_{t+1}/C_t) = g(X_t,X_{t+1},Y_{t+1})$ where the conditional distribution of $(X_{t+1},Y_{t+1})$ given $(X_t,Y_t)$ depends only on $X_t$ by redefining the state as $(X_t,Y_t)$.} Then $(1-\rho)\log(V_t/C_t) =: v(X_t)$, where $v$ solves  
\begin{equation} \label{e:ez} 
 v(x) =  \log \left( (1-\beta) + \beta  \mb E^Q\left[ \left. e^{ \kappa v(X_{t+1}) + (1-\gamma) g(X_t,X_{t+1}) }  \right| X_t  = x \right]^\frac{1}{\kappa} \right) 
\end{equation}
with $\kappa = \frac{1-\gamma}{1-\rho}$ (see, e.g., \cite{HHL}). The properties of this recursion are different for $\kappa < 0$, $\kappa \in (0,1)$, and $\kappa \in [1,\infty)$. We focus on the case $\kappa < 0$, as it is the pertinent case in the long-run risks literature where typically $\gamma > 1$ and $1/\rho > 1$. 

\subsection{Existing results}

\cite{EZ} and \cite{MarinacciMontrucchio2010} derived sufficient conditions for existence and uniqueness when consumption growth is bounded. \cite{AlvarezJermann2005} establish existence and uniqueness when consumption growth is i.i.d. with bounded innovations. \cite{GuoHe} establish sufficient conditions for existence and uniqueness with finite state space. \citeauthor{BorovickaStachurski} (\citeyear{BorovickaStachurski}; BS hereafter) present necessary and sufficient conditions for existence when $\mc X$ is compact (under additional side conditions on $Q$). Our results below and those of BS are non-nested if $\mc X$ is compact: we do not impose any side conditions on $Q$, but we also do not establish uniqueness in the compact case.

\citeauthor{HS2012} (\citeyear{HS2012}; HS hereafter) and BS establish existence with unbounded $\mc X$ when $\kappa < 0$.\footnote{\cite{HS2012} and \cite{RenStachurski} establish uniqueness when $\kappa \geq 1$.} 
We also only present sufficient conditions for existence because the operator does not have a subgradient of the form studied in Section \ref{s:sr}. Connections between our conditions and those in HS and BS are discussed in more detail below.

\subsection{New results}

Under general conditions (see \cite{HansenScheinkman2009} and \cite{Christensen2015,npsdfd}), there exists a strictly positive function $\iota$ and scalar $\lambda > 0$ solving\footnote{Note the function $\iota$ is defined only up to scale normalization.} the equation
\begin{equation} \label{e:efn}
 \lambda \iota(x) = \mb E^Q[ \iota(X_{t+1}) (C_{t+1}/C_t)^{1-\gamma} |X_t = x] \,.
\end{equation} 
\cite{HansenScheinkman2009} use $\iota$ and $\lambda$ to define a distorted conditional expectation operator
\[
 \tilde{\mb E} f(x) =  \mb E^Q \left[ \left.  \frac{\iota(X_{t+1}) (C_{t+1}/C_t)^{1-\gamma}}{\lambda \iota(X_t)} f(X_{t+1}) \right| X_t = x \right]  \,.
\]
HS show that solving (\ref{e:ez}) is equivalent to finding a fixed point of
\begin{equation} \label{e:T-ez-1}
 \mb T f(x) = \log \left( (1-\beta) \iota(x)^{-\frac{1}{\kappa}} + \beta \lambda^{\frac{1}{\kappa}} \tilde{\mb E}[ e^{\kappa f(X_{t+1})} | X_t = x]^{\frac{1}{\kappa}} \right) \,,
\end{equation}
with the solution to recursion (\ref{e:ez}) and the fixed point of $\mb T$ differing additively by $\frac{1}{\kappa} \log \iota$.\footnote{The version of recursion (\ref{e:efn}) above appears on p. 11968 of HS. In our notation, their recursion is $\hat{\mb U}g(x) = (1-\beta) \iota(x)^{-\frac{1}{\kappa}} + \beta \lambda^{\frac{1}{\kappa}} \tilde{\mb E}[ g(X_{t+1})^\kappa | X_t = x]^{\frac{1}{\kappa}}$. Recursion (\ref{e:efn}) is obtained by setting $\mb T f = \log (\hat{\mb U} (\exp(f)))$.} 

We follow HS and assume $X$ is stationary under the law of motion corresponding to the distorted conditional expectation $\tilde{\mb E}$. Let $\tilde \mu$ denote the stationary distribution induced by $\tilde{\mb E}$ and let $\tilde E^{\phi_r}$ denote the corresponding Orlicz heart defined using $\tilde \mu$. 
Our first regularity condition requires that $\log \iota$ has thin tails, in the sense that
\begin{equation} \label{e:ez-phi-1}
 \log \iota \in \tilde E^{\phi_r} \quad \mbox{for some $r \geq 1$.}
\end{equation}
Under this condition, Lemma \ref{lem:T-prop-ez} shows that $\mb T$ is a continuous, monotone operator on $\tilde E^{\phi_s}$ for each $1 \leq s \leq r$. It is clear that $\mb T v \geq \log ( (1-\beta) \iota(x)^{-\frac{1}{\kappa}})$.  Therefore, should there exist a $\bar v \in \tilde E^{\phi_r}$ for which $\mb T \bar v \leq \bar v$, the sequence of iterates $\mb T^n \bar v$ must be bounded from below. The remainder of the proof shows that the inequality $\mb T \bar v \leq \bar v$ holds for the function
\[
 \bar v(x) = \log \left(  (1-\beta)\sum_{n=0}^\infty (\beta \lambda^{\frac{1}{\kappa}})^n \tilde{\mb E}^n(\iota^{-\frac{1}{\kappa}})(x) \right) \,.
\]
The sum is convergent under the  eigenvalue condition from \cite{HS2012}:
\begin{equation} \label{e:eval}
 \beta \lambda^{\frac{1}{\kappa}} < 1 \,.
\end{equation}

\begin{remark}
Although $\mb T$ is not contractive, it follows from Proposition \ref{p:exun}(i) that the sequence of iterates $\bar v, \mb T \bar v, \mb T^2 \bar v, \ldots$ will converge to a fixed point of $\mb T$ under the conditions of Theorem \ref{t:ez-exist-1} and Corollary \ref{c:ez-exist-2} below. The same is true for the sequence of iterates $\ul v, \mb T \ul v, \mb T^2 \ul v, \ldots$ with $\ul v (x) = \log(1-\beta) - \kappa^{-1} \log \iota(x)$.
\end{remark}

\begin{theorem}\label{t:ez-exist-1}
Let $X$ be stationary under the law of motion corresponding to the distorted conditional expectation $\tilde{\mb E}$, $\kappa < 0$, and conditions (\ref{e:ez-phi-1}) and (\ref{e:eval}) hold. Then: $\mb T$ has a fixed point in $\tilde E^{\phi_s}$ and therefore the recursion (\ref{e:ez}) has a solution $v \in \tilde E^{\phi_s}$ for all $s \in [1,r]$.
\end{theorem}

Condition (\ref{e:eval}) is the eigenvalue condition under which HS establish existence in $L^1(\tilde \mu)$. BS showed this condition is necessary for existence (under some additional operator-theoretic side conditions). 
Condition (\ref{e:ez-phi-1}) is stronger than the integrability conditions imposed on $\iota$ in Assumptions 4 and 5 of HS. However, this condition does not seem to bite for models commonly encountered (see the linear-Gaussian example below) and also ensures that the stochastic discount factor (SDF)
\begin{equation} \label{e:sdf}
 \beta {\left(\frac{C_{t+1}}{C_t}\right)\!\!}^{-\rho} \left[ \frac{V_{t+1}^{1-\gamma}}{\mb E^Q[V_{t+1}^{1-\gamma}|\mathcal F_t]} \right]^\frac{\rho-\gamma}{1-\gamma} \equiv \beta e^{-\rho g(X_t,X_{t+1})}\left[ \frac{e^{ \kappa v(X_{t+1}) + (1-\gamma) g(X_t,X_{t+1}) } }{\mb E^Q[e^{ \kappa v(X_{t+1}) + (1-\gamma) g(X_t,X_{t+1}) } |X_t]} \right]^\frac{\rho-\gamma}{1-\gamma}
\end{equation}
is well defined provided consumption growth $g$ has sufficiently thin tails.

Theorem \ref{t:ez-exist-1} has implications for existence in spaces defined relative to the stationary distribution $\mu$ of $X$. Suppose that $\tilde \mu$ and $\mu$ are mutually absolutely continuous and let $\Delta = \frac{\mr d \tilde \mu}{\mr d \mu}$ denote the change of measure of $\tilde \mu$ with respect to $\mu$. Consider the thin-tail condition
\begin{equation} \label{e:rn}
 \mb E^\mu[\Delta(X_t)^{1 + \varepsilon}] < \infty \quad \mbox{and} \quad \mb E^\mu[\Delta(X_t)^{-\varepsilon}] < \infty \quad \mbox{for some $\varepsilon > 0$.}
\end{equation} 
A sufficient condition for (\ref{e:rn}) is that $\log \Delta \in L^{\phi_1}$. The spaces $\tilde E^{\phi_r}$ (defined using $\tilde \mu$) and $E^{\phi_r}$ (defined using $\mu$) are equivalent under condition (\ref{e:rn}); see Lemma \ref{lem:embed:0}. We may therefore restate condition (\ref{e:ez-phi-1}) as
\begin{equation} \label{e:ez-phi-2}
 \log \iota \in E^{\phi_r} \quad \mbox{for some $r \geq 1$.}
\end{equation}

\begin{corollary}\label{c:ez-exist-2}
Let $X$ be stationary under the law of motion corresponding to the distorted conditional expectation $\tilde{\mb E}$, $\kappa < 0$, and conditions (\ref{e:eval}), (\ref{e:rn}), and (\ref{e:ez-phi-2}) hold. Then: $\mb T$ has a fixed point in $E^{\phi_s}$ and therefore the recursion (\ref{e:ez}) has a solution $v \in E^{\phi_s}$ for all $s \in [1,r]$.
\end{corollary}

\medskip

\paragraph{Example: Linear-Gaussian environments.} Consider an environment studied in Section I.A of \cite{BansalYaron2004}, \cite{HHL}, and \cite{BKSY2014}, amongst others, where $X$ evolves as a stationary Gaussian VAR(1):
\[
 X_{t+1} = \nu + AX_t + u_{t+1}\,, \quad u_t \sim N(0,\Sigma)\,,
\]
with all eigenvalues of $A$ inside the unit circle, and $g(X_t,X_{t+1}) = \delta'X_{t+1}$ 
for some vector $\delta$ (this is trivially true if log consumption growth is itself a component of $X_t$). Solving (\ref{e:efn}), 
\begin{align*}
 \iota(x) & = e^{ (1-\gamma) \delta' A (I - A)^{-1} x} \,, & 
 \lambda & = e^{ \frac{(1-\gamma)^2}{2}\delta'(I-A)^{-1} \Sigma (I-A')^{-1} \delta + (1-\gamma) \delta'(I-A)^{-1} \nu } \,.
\end{align*}
To apply Corollary \ref{c:ez-exist-2} we must verify conditions (\ref{e:eval}), (\ref{e:rn}), and (\ref{e:ez-phi-2}). 
To verify condition (\ref{e:rn}), first note 
\[
 \frac{\iota(X_{t+1}) (C_{t+1}/C_t)^{1-\gamma}}{\lambda \iota(X_t)} = e^{(1-\gamma)\delta'(I-A)^{-1} u_{t+1} - \frac{(1-\gamma)^2}{2}\delta'(I-A)^{-1}\Sigma(I-A')^{-1} \delta}
\]
so the $u_t$ are i.i.d. $N( (1-\gamma)\delta'(I-A)^{-1}\Sigma, \Sigma)$ under $\tilde{\mb E}$. Equivalently, under $\tilde{\mb E}$ we have
\[
 X_{t+1} = \nu + (1-\gamma)\delta'(I-A)^{-1}\Sigma + AX_t + u_{t+1}\,, \quad u_t \sim N(0,\Sigma)\,.
\]
This implies the stationary distributions $\mu$ and $\tilde \mu$ are both Gaussian, with different means but the same covariance. In consequence, $\log \Delta(x)$ is affine in $x$ and so condition (\ref{e:rn}) holds for any $\varepsilon > 0$. As $\log \iota(x)$ is also affine in $x$, we have that $\log \iota \in E^{\phi_r}$ for all $r \in [1,2)$, which verifies condition (\ref{e:ez-phi-2}).
It follows that the single condition one needs to verify for existence of recursive utilities in linear-Gaussian environments is the eigenvalue condition (\ref{e:eval}), which reduces to
\[
 \beta e^{ \frac{(1-\rho)(1-\gamma)}{2}\delta'(I-A)^{-1} \Sigma (I-A')^{-1} \delta + (1-\rho) \delta'(I-A)^{-1} \nu } < 1  \,.
\]
Note also that as $g(X_t,X_{t+1}) = \delta' X_{t+1}$, which belongs to $E^{\phi_r}$ for $r \in [1,2)$, the SDF (\ref{e:sdf}) is therefore well defined and all of its moments exist. \hfill $\square$

\appendix

\section{Proofs}\label{s:proofs}

\begin{remark}\label{rmk:r3}
Several of the proofs below require showing that a function $f$ is an element of $E^{\phi_s}$ with $s \geq 1$. That is, that $\mb E^\mu[\exp(|f(X_t)/c|^s)] < \infty$ holds for all $c > 0$. For any $0 < \bar c < c$ we have $(\bar c/c)^s < 1$ and therefore
\[
 \exp(|f(X_t)/c|^s) = (\exp(|f(X_t)/\bar c|^s))^{(\bar c/c)^s} \leq \exp(|f(X_t)/\bar c|^s)
\]
because $\exp(|f(X_t)/\bar c|^s) \geq 1$. 
In order to show that $f \in E^{\phi_s}$, one therefore only has to check that $\mb E^\mu[\exp(|f(X_t)/c|^s)]< \infty$ holds for all $c \in (0,\epsilon)$ for any fixed $\epsilon > 0$.  
\end{remark}

\subsection{Ancillary results}

A version of this first Lemma appears in Chapter 2.3 of the manuscript \cite{Pollard} and is used frequently to control the Orlicz norm $\|\cdot\|_{\phi_r}$. We include a proof for convenience. 

\begin{lemma}[\cite{Pollard}]\label{lem:pollard}
Let $\mb E^{\mu}[ \exp(|f(X)/C|^r) ] - 1 \leq C'$ for finite constants $C > 0$ and $C' \geq 1$. Then: $\|f\|_{\phi_r} \leq CC'$.
\end{lemma}

\begin{proof}[Proof of Lemma \ref{lem:pollard}]
Take $\tau \in [0,1]$. By convexity of $\psi(x) := \exp(|x|^r) - 1$, we have
\[
 \mb E^{\mu}[ \psi(\tau|f(X)|/C) ] \leq \tau \mb E^{\mu}[ \psi(|f(X)|/C) ] + (1-\tau) \psi(0) = \tau \mb E^{\mu}[ \psi(|f(X_t)|/C) ] \,.
\]
The result follows by setting $\tau = 1/C'$.
\end{proof}

\begin{lemma}[\cite{Karakostas,Chen2016}] \label{lem:holder-infinite}
Let $1 < p_i < \infty$ for $i \in \mb N$, and $\sum_{i=1}^\infty \frac{1}{p_i} = 1$. If $\prod_{i=1}^\infty \|f_i\|_{p_i} < \infty$ then $\prod_{i=1}^\infty f_i$ is well defined and $\|\prod_{i=1}^\infty f_i \|_1 \leq \prod_{i=1}^\infty \|f_i\|_{p_i}$.
\end{lemma}

Let $\mu$ and $\nu$ be two probability measures on a measurable space $(\mc X,\mcr X)$. We make explicit the dependence of function classes and norms on the measures $\mu$ and $\nu$. Let $\Delta = \frac{\mr d \mu}{\mr d \nu}$, and let $\| \Delta\|_{L^p(\nu)}$ denote its $L^p(\nu)$ norm. 

\begin{lemma}\label{lem:embed:0}
Let $\mu \ll \nu$ and $\int \Delta^p \, \mr d \nu < \infty$ for some $p > 1$. Then: $E^{\phi_r}(\nu) \hookrightarrow E^{\phi_r}(\mu)$ and $L^{\phi_r}(\nu) \hookrightarrow L^{\phi_r}(\mu)$ for each $r \geq 1$.
\end{lemma}

\begin{proof}[Proof of Lemma \ref{lem:embed:0}]
To see that $E^{\phi_r}(\nu) \subseteq E^{\phi_r}(\mu)$, take any $f \in E^{\phi_r}(\nu)$ and $c > 0$. Then:
\[
 \mb E^{\mu} \left[ e^{|f(X)/c|^r}\right] = \mb E^{\nu} \left[\Delta (X) e^{|f(X)/c|^r} \right] \leq \|\Delta\|_{L^p(\nu)} \mb E^{\nu} \left[ e^{\left|f(X)/(c/q^{1/r})\right|^r} \right]^{\frac{1}{q}} < \infty \,,
\]
where $q > 1$ is the dual index of $p$. Therefore, $f \in E^{\phi_r}(\mu)$. Similarly, $L^{\phi_r}(\nu) \subseteq L^{\phi_r}(\mu)$.

For continuity of the embedding, take $f \in L^{\phi_r}(\nu)$ and $c = q^\frac{1}{r}\|f\|_{\phi_r(\nu)}$. Substituting into the above display yields
\[
 \mb E^{\mu} [ e^{|f(X)/c|^r}] \leq 2^\frac{1}{q} \|\Delta\|_{L^p(\nu)} \,.
\]
Therefore, $\|f\|_{L^{\phi_r}(\mu)} \leq ((2^\frac{1}{q}  \| \Delta\|_{L^p(\nu)}  -1 ) \vee 1) \|f\|_{L^{\phi_r}(\nu)}$ by Lemma \ref{lem:pollard}.
\end{proof}

\subsection{Proofs for Section \ref{sec:problems}}

\begin{proof}[Proof of Proposition \ref{prop:ssy-nonexist}]
Suppose a solution $v \in L^1$ to (\ref{eq:ssy_recur}) does indeed exist for some $\alpha \neq 0$. Then $v$ is a fixed point the operator $\mb T$. Consider the related operator $\mb S$, given by
\[
 \mb S f (h) = \mathsf a + \mathsf b e^{2h}  + \beta \mb E^{Q} \left[ \left. f(h_{t+1}) \right| h_t = h \right] .
\]
As $\mb S$ is a contraction mapping on $L^1$, we may deduce it has a unique fixed point $w \in L^1$ given by
\[
 w(h) = \frac{\mathsf a}{1-\beta} + \mathsf b \sum_{i=0}^\infty \beta^{i} \mb E^Q[ e^{2h_{t+i}} | h_t = h] \,. 
\]
By Jensen's inequality, $\mb T f \geq \mb S f$ holds for any $f$. Note $w - v = \mb S w - \mb T v \leq \mb S w - \mb S v$, where $\mb S w(h) - \mb S v (h) = \beta \mb E^Q[w(h_{t+1}) - v(h_{t+1})|h_t = h] =: \mb D(w - v)(h)$. Therefore, $(\mb I - \mb D)(w - v) \leq 0$. As $(\mb I - \mb D)$ is invertible on $L^1$ (see the discussion in Section \ref{s:sr}) and its inverse maps non-negative functions to non-negative functions, we have $w - v \leq 0$ and hence that $v \geq w$. Also note that $w \geq \ul w$, where
\[
 \ul w(h) = \frac{\mathsf a}{1-\beta} + \mathsf b e^{2h} \,.
\] 
By monotonicity and the fact that the fixed point $v$ of $\mb T$ is bounded below by $\ul w$, we have
\begin{equation}
 v = \mb T v \geq \mb T \ul w \,, \label{eq:ssy_non}
\end{equation}
where
\[
 \mb T \ul w (h) = \mathsf a + \mathsf b e^{2h} + \beta \log \mb E^{Q} \left[ \left. \exp \left(\frac{\mathsf a}{1-\beta} + \mathsf b e^{2h_{t+1}} \right) \right| h_t = h \right] .
\]
But note that the right-hand side expectation is $+\infty$ for every $h$ because $\mathsf b > 0$. It follows by inequality (\ref{eq:ssy_non}) that $v(h) = +\infty$ almost everywhere, which contradicts $v \in L^1$.
\end{proof}

\begin{proof}[Proof of Proposition \ref{prop:bs-nonunique}]
Substituting $v(h) = a + b h$ into (\ref{eq:bs_recur}) and using the conditional characteristic function for the autoregressive gamma process \cite[Appendix H]{BCZ}, we obtain
\[
 a + bh = \mathsf a + \mathsf b h + \beta a + \frac{ \beta \varphi b }{1-b c} h - \beta \delta \log (1 - b c) \,.
\]
Matching coefficients gives a quadratic equation in $b$. When $\mathsf q := 1 + c \mathsf b - \beta \varphi$ satisfies $\mathsf q^2 - 4 c \mathsf b > 0$, there are two solutions for $b$:
\[
 b_1 = \frac{\mathsf q - \sqrt{\mathsf q^2 - 4 c \mathsf b}}{2 c} \,, \quad \quad b_2 = \frac{\mathsf q + \sqrt{\mathsf q^2 - 4 c \mathsf b}}{2 c} \,,
\]
both of which satisfy $1 - b c > 0$. Therefore, there are two solutions of the form $v_i(h) = a_i + b_i h$, where $a_i = \frac{\mathsf a - \beta \delta \log(1 - b_i c)}{1-\beta}$, $i = 1,2$.
\end{proof}

\subsection{Proofs for Section \ref{s:preliminaries}}

\begin{proof}[Proof of Proposition \ref{p:exun}]
Existence: we prove this for case (a); similar arguments apply for (b). 
The sequence $\{\bar v_n\}_{n \geq 1}$ with $\bar v_n = \mb T^n \bar v$ is monotone and bounded below by $\ul v$. It follows by the monotone convergence property that $\{\bar v_n\}_{n \geq 1}$ converges to some $v \in \mc E$ with $v \geq \ul v$. Finally, $\| \mb Tv - v\| \leq \|\mb Tv - \mb T \bar v_n\| + \| \mb T \bar v_n - v\| = \|\mb T v - \mb T \bar v_n\| + \| \bar v_{n+1} - v\| \to 0$ by continuity of $\mb T$, hence $\mb T v = v$.

Uniqueness: Suppose $\mb T$ satisfies (\ref{e:subgrad}) at each fixed point. Let $v,v' \in \mc E$ be fixed points of $\mb T$. By (\ref{e:subgrad}), we have $v' - v = \mb T v' - \mb T v \geq \mb D_{v}(v' - v)$, which implies that 
\begin{equation} \label{e-idpf-ineq}
 (\mb I - \mb D_{v})(v' - v) \geq 0 \,.
\end{equation}
As $\rho(\mb D_v;\mc E) < 1$, we have $(\mb I - \mb D_v)^{-1} = \sum_{i=0}^\infty (\mb D_v)^i$ where the series converges in operator norm \cite[Theorem 10.15]{Kress}. The operator $\mb D_v$ is monotone and so $(\mb I - \mb D_v)^{-1}$ is also monotone. Applying $(\mb I - \mb D_v)^{-1}$ to both sides of equation (\ref{e-idpf-ineq}) yields  $v'-v \geq 0$.
A parallel argument yields $v - v' \geq 0$. Therefore, $v = v'$. 
The proof follows by parallel arguments when $\mb T$ instead satisfies (\ref{e:supergrad}) at each of its fixed points.
\end{proof}

\begin{proof}[Proof of Corollary \ref{cor:nonunique-gen}]
 Suppose $\mb T$ satisfies (\ref{e:subgrad}) at each fixed point. By (\ref{e:subgrad}), for $v,v' \in \mc V$:
\[
 v'-v = \mb T v' - \mb Tv \geq \mb D_v(v'-v)
\]
hence $(\mb I - \mb D_v)(v'-v) \geq 0$. When $\rho(\mb D_v;\mc E) < 1$, the operator $(\mb I - \mb D_v)$ is invertible on $\mc E$ with $(\mb I - \mb D_v) ^{-1} = \sum_{n=0}^\infty \mb D_v^n$. As $\mb D_v$ is monotone, so too is $(\mb I - \mb D_v)^{-1}$. Applying $(\mb I - \mb D_v)^{-1}$ to both sides of the above display yields $v' - v \geq 0$, so $v$ is the smallest fixed point of $\mb T$. 

Suppose any other $v' \in \mc V$ distinct from $v$ were also stable. Then we could apply an identical argument to obtain the reverse inequality $v - v' \geq 0$, a contradiction.
The proof when $\mb T$ satisfies (\ref{e:supergrad}) at each fixed point follows similarly.
\end{proof}

\begin{lemma}\label{lem:cge-nbhd}
Let $v \in \mc E$ be a stable fixed point of $\mb T$, and let there exist a neighborhood $N$ of $v$ for which
\begin{equation} \label{eq:frechet}
 \mb T f - \mb T v = \mb D_v (f - v) + o(\|f - v\|)
\end{equation}
for all $f \in N$. Then: there exists a neighborhood $N'$ of $v$ for which $\lim_{n \to \infty} \mb T^n f = v$ for all $f \in N'$
\end{lemma}

\begin{proof}[Proof of Lemma \ref{lem:cge-nbhd}]
As $\rho(\mb D_v; \mc E) < 1$, there exists $n_0 \in \mb N$ and $\epsilon > 0$ for which $\|(\mb D_v)^{n_0} f\| \leq e^{-\epsilon n_0} \|f\|$ for all $f \in \mc E$. Recursively applying condition (\ref{eq:frechet}), we may deduce that there exists a sufficiently small neighborhood $N'$ of $v$ upon which
\begin{equation} \label{eq:frechet-1}
 \mb T^n f - v = (\mb D_v)^n (f - v) + o( \|f - v\|)  \,, \quad \mbox{ for all } 1 \leq n \leq n_0 \,,
\end{equation}
and hence
\[
 \| \mb T^{n_0} f - v \| \leq e^{-\epsilon n_0} \| f - v \| + o( \|f - v\|) \,.
\]
Making $N'$ smaller if necessary, we may therefore deduce that there is a $\varrho \in (0,1)$ for which $\| \mb T^{n_0} f - v\| \leq \varrho \| f - v\|$ holds for all $f \in N'$. For any $f \in N'$ and $k \in \mb N$, we therefore have that $\|\mb T^{k n_0} f - v\| \leq \varrho^{k} \| f - v\|$. Moreover, for any $n \in \mb N$ that is not an integer multiple of $n_0$, it follows by (\ref{eq:frechet-1}) with $k = \lfloor n/n_0 \rfloor$ that $\mb T^n f - v = \mb T^{n - k n_0}( \mb T^{k n_0} f ) -v = (\mb D_v)^{n - k n_0}( \mb T^{k n_0} f - v) + o( \| \mb T^{k n_0} f - v\|) = O( \| \mb T^{k n_0} f - v\|) = O(\varrho^k)$.
\end{proof}

\begin{proof}[Proof of Corollary \ref{cor:cge}]
Suppose condition (a) holds.  Fix $w \in \mc E$ with $w \leq \bar v$, let $w_0 = w$, and let $w_n = \mb T^n w$ for $n \in \mb N$. Also let $\bar v_n = \mb T^n \bar v$. By Proposition \ref{p:exun} we know that there is a unique fixed point $v \in \mc E$. Then by monotonicity of $\mb T$ and the subgradient inequality (\ref{e:subgrad}), for every $n \in \mb N$ we have
\[
 \bar v_n - v \geq w_n - v = \mb T w_{n-1} - \mb T v \geq \mb D_v (w_{n-1} - v) \geq (\mb D_v)^n (w - v)\,,
\]
where the final inequality is by monotonicity of $\mb D_v$. The left-hand side term $\bar v_n - v \to 0$ as $n \to \infty$ by Proposition \ref{p:exun}. Moreover, as $\rho(\mb D_v; \mc E) < 1$, there exists $n_0 \in \mb N$ and $\epsilon > 0$ for which $\|(\mb D_v)^{n_0} f\| \leq e^{-\epsilon n_0} \|f\|$ for all $f \in \mc E$, from which we may deduce that the right-hand side term $ (\mb D_v)^n (w - v) \to 0$ as $n \to \infty$. As $\|\cdot\|$ is a lattice norm, it follows that $w_n \to v$ as $n \to \infty$. 
The proof when (b) holds and $\mb T$ satisfies (\ref{e:supergrad}) follows similarly.
\end{proof}

\begin{lemma}\label{lem:mcp}
Let $\mu$ be a probability measure on $(\mc X, \mcr X)$. Then: for any $r \geq 1$, the space $E^{\phi_r}$ has the monotone convergence property.
\end{lemma}

\begin{proof}[Proof of Lemma \ref{lem:mcp}]
Let $\{f_n\}_{n \geq 1} \subset E^{\phi_r}$ be an increasing sequence of functions bounded above by some $g \in E^{\phi_r}$. As $E^{\phi_r} \hookrightarrow L^1(\mu)$, the sequence $\{f_n\}_{n \geq 1}$ is uniformly bounded in $L^1(\mu)$ and so it follows by Beppo Levi's monotone convergence theorem \cite[Theorem I.7.1]{Malliavin} that there exists $f \in L^1(\mu)$ for which $\lim_{n \to \infty} f_n = f$ ($\mu$-almost everywhere) and $\lim_{n \to \infty} \|f_n - f\|_1 = 0$, where $\|\cdot\|_1$ denotes the $L^1(\mu)$ norm. As $f_1 \leq f \leq g$, we have $|f| \leq |f_1| + |g|$. Moreover, as $f_1,g \in E^{\phi_r}$, for any $c > 0$ we have
\begin{align*}
 \mb E^{\mu}[ \exp(|f(X)/c|^r)] 
 & \leq \mb E^{\mu}[ \exp(((|f_1(X)| + |g(X)|)/c)^r)] \\
 & \leq \frac{1}{2} \mb E^{\mu}[ \exp(|2f_1(X)/c|^r)] + \frac{1}{2} \mb E^{\mu}[ \exp(|2g(X)/c|^r)] < \infty\,,
\end{align*}
from which it follows that $f \in E^{\phi_r}$. 

To establish convergence in $\|\cdot\|_{\phi_r}$,  suppose that $\limsup_{n \to \infty} \| f_n - f\|_{\phi_r} \geq 2\varepsilon$ for some $\varepsilon > 0$. Then
\begin{equation}\label{e-idpf-norm}
 \limsup_{n \to \infty} \mb E^\mu[ \exp(|(f_n(X) - f(X))/\varepsilon|^r) ] \geq 2\,.
\end{equation}
Note that $\{g_n\}_{n \geq 1}$ with $g_n = \exp(|(f_n - f)/\varepsilon|^r)$ is a monotone sequence of non-negative functions with $\limsup_{n \to \infty}g_n = 0$ ($\mu$-almost everywhere). Moreover, for each $n \geq 1$ we have that 
\[
 g_n \leq \exp \big(((|f_1| + |g| + |f|)/\varepsilon)^r \big) \,,
\]
and the right-hand side is $\mu$-integrable because $f_1, g, f \in E^{\phi_r}$. Therefore, by reverse Fatou:
\[
 \limsup_{n \to \infty} \mb E^\mu[ \exp(|(f_n(X) - f(X))/\varepsilon|^r) ] 
 \leq  \mb E^\mu \big[ \limsup_{n \to \infty} \exp(|(f_n(X) - f(X))/\varepsilon|^r) \big] = 0 \,,
\]
contradicting (\ref{e-idpf-norm}). It follows that $\|f_n - f\|_{\phi_r} \to 0$.
\end{proof}

\begin{remark}
It follows by an identical argument to Lemma \ref{lem:mcp} that the Orlicz heart $E^\psi := \{f \in L^0 : \mb E^\mu[ \psi(f(X)/c)] < \infty$ for all $c > 0\}$ defined using any monotone, continuously differentiable, strictly convex $\psi : \mb R_+ \to \mb R_+$ with $\psi(0)$ and $\lim_{x \to \infty} \psi(x)/x \to +\infty$ has the monotone convergence property when equipped with the corresponding Luxemburg norm $\|f\|_{\psi} := \inf \left\{ c > 0 : \mb E^\mu[\psi(|f(X)/c|)] \leq 1 \right\}$. 
\end{remark}

We next present an intermediate result used to prove Lemma \ref{lem:com}. Note that condition (\ref{e:m-thin}) implies that $ (\log m \vee 0) \in L^{\phi_r}(\mu \otimes Q)$, the Orlicz class of functions $f : \mc X \times \mc X \to \mb R$ defined relative to the stationary distribution $\mu \otimes Q$ of $(X_t,X_{t+1})$. With slight abuse of notation, let $\|(\log m \vee 0)\|_{\phi_r}$ denote the corresponding Orlicz norm of $(\log m \vee 0)$.

\begin{lemma}\label{lem:msubexp}
Let $\tilde{\mb E}$ be of the form (\ref{e:rn-sufficient}) and let $m$ satisfy condition (\ref{e:m-thin}). Then for any $p \in (1,\infty)$ and $n \geq 1$: 
\[
 \mb E^{\mu\otimes Q}[m(X_t,X_{t+1})^{np}]^{1/p} \leq  e^{(2 n \|(\log m \vee 0)\|_{\phi_r})^\frac{r}{r-1}(2 p)^\frac{1}{r-1}} + 2^\frac{3}{2p} \,.
\]
Moreover, for any $\beta \in (0,1)$ there exists $C \in (0,\infty)$ and $c \in (0,1-\beta)$ depending only on $\beta$, $r$, $\|(\log m \vee 0)\|_{\phi_r}$, and $p$ such that the inequality
\[
 \mb E^{\mu\otimes Q}[m(X_t,X_{t+1})^{np}]^{1/p} \leq C e^{(\beta+c)^{-n}}
\]
holds for each $n \geq 1$.
\end{lemma}

\begin{proof}[Proof of Lemma \ref{lem:msubexp}]
First note $\mb E^{\mu\otimes Q}[m(X_t,X_{t+1})^{np}] \leq \mb E^{\mu\otimes Q}[ e^{np|\!\log m(X_t,X_{t+1}) \vee 0|}]$. To simplify notation, let $Y_t = (X_t,X_{t+1})$, $a = \log m \vee 0$, and $\|a\|_{\phi_r} = \|(\log m \vee 0)\|_{\phi_r}$. In what follows, all probabilities (denoted $\Pr(\cdot)$) are taken with respect to  $\mu \otimes Q$.  Let $A$ be a positive constant (specified below) and set $|a| = a_+ + a_-$ with $a_+ = |a| \ind\{|a| \leq A\}$ and $a_- = |a| \ind\{|a| > A\}$. For any $z > 0$, we have
\begin{align} \label{e:tail}
  \Pr\left( e^{np|a(Y_t)|} \geq z \right) & \leq \Pr\left( a_+(Y_t)  \geq \frac{\log z}{2np}  \right)  + \Pr\left( a_-(Y_t)  \geq \frac{\log z}{2np} \right) \,.
\end{align}

By Markov's inequality and definition of $\|\cdot\|_{\phi_r}$, we have 
\begin{align*}
  \Pr\left( a_-(Y_t)  \geq \frac{\log z}{2np} \right) 
  & \leq \Pr\left( |a(Y_t)|^r \geq \frac{A^{r-1} \log z}{2np} \right) \\
  & = \Pr\left( \exp \left( \frac{|a(Y_t)|^r}{\|a\|_{\phi_r}^r} \right)  \geq \exp \left(\frac{1}{\|a\|_{\phi_r}^r} \frac{A^{r-1} \log z}{2np} \right) \right) \\
  & \leq  \frac{\mb E^{\mu \otimes Q}\left[ \exp \left( \left| a(Y_t)/\|a\|_{\phi_r} \right|^r \right) \right] }{\exp \left( \frac{1}{\|a\|_{\phi_r}^r} \frac{A^{r-1} \log z}{2np} \right)} \\
  & \leq  2 \exp \left( -\frac{1}{\|a\|_{\phi_r}^r} \frac{A^{r-1} \log z}{2np} \right) \,.
\end{align*}
Setting $A = ( \|a\|_{\phi_r}^r 4np)^\frac{1}{r-1}$, we obtain
\[
 \Pr\left( a_-(Y_t)  \geq \frac{\log z}{2np} \right) 
 \leq 2 z^{- 2 } \,.
\]
As $2z^{-2} \geq 1$ if $z \leq \sqrt 2$, we therefore have
\begin{equation} \label{e:tail:trunc}
 \int_0^\infty \Pr\left( a_-(Y_t)  \geq \frac{\log z}{2np} \right)  \, \mr d z \leq \sqrt 2 + 2 \int_{\sqrt 2}^\infty z^{-2} \, \mr d z = 2^\frac{3}{2}  \,.
\end{equation}

For the first term on the right-hand side of (\ref{e:tail}), as $a_+ \leq A$ we have 
\begin{equation} \label{e:tail:trunc-2}
 \Pr\left(  a_+(Y_t)  \geq \frac{\log z}{2np} \right)  = 0 \mbox{ if } z > e^{2np A}\,.
\end{equation}
Note $2np A = (2 n p \|a\|_{\phi_r} )^\frac{r}{r-1}2^\frac{1}{r-1}$.  Using the fact that $\mb E[Z] = \int_0^\infty \Pr(Z \geq z)\, \mr dz$ for a non-negative random variable $Z$, we may deduce from (\ref{e:tail}), (\ref{e:tail:trunc}), and  (\ref{e:tail:trunc-2}) that
\begin{align*}
 \mb E^{\mu\otimes Q}[m(X_t,X_{t+1})^{np}]
 & \leq \int_0^\infty \Pr( e^{np |a(Y)|} \geq z) \, \mr d z  \leq e^{(2 n p \|a\|_{\phi_r})^\frac{r}{r-1}2^\frac{1}{r-1}} + 2^\frac{3}{2} \,.
\end{align*}
The first assertion follows because $(x+y)^{1/p} \leq x^{1/p} + y^{1/p}$ for $x,y \geq 0$ and $p \geq 1$.
The second assertion follows as $n^{\frac{r}{r-1}} = o( (\beta + c)^{-n})$ for any $\beta \in (0,1)$ and $c \in (0,1-\beta)$.
\end{proof}

\begin{proof}[Proof of Lemma \ref{lem:com}]
We first show $\mb D$ is a bounded linear operator on $L^{\phi_s}$ for any $s \geq 1$. Linearity follows by inspection. For boundedness, fix any $s \geq 1$ and take any $f \in L^{\phi_s}$ with $\|f\|_{\phi_s} > 0$ and any $q \in (0,1)$. By applying Jensen's inequality, definition of $\tilde{\mb E}$ from (\ref{e:rn-sufficient}), and H\"older's inequality with $p^{-1} + q^{-1} = 1$, we obtain
\begin{align*}
 \mb E^{\mu}\left[ e^{|\mb D f(X_t)/(q^\frac{1}{s} \beta \|f\|_{\phi_s})|^s} \right] 
 & = \mb E^{\mu}\left[ e^{q^{-1} |\tilde{\mb E} f(X_t)/\|f\|_{\phi_s}|^s} \right] \\
 & \leq \mb E^{\mu \otimes Q}\left[  m(X_t,X_{t+1}) e^{q^{-1} | f(X_{t+1})/\|f\|_{\phi_s}|^s} \right] \\
 & \leq \mb E^{\mu \otimes Q}\left[  m(X_t,X_{t+1})^p \right]^\frac{1}{p} \mb E^\mu \left[ e^{| f(X_t)/\|f\|_{\phi_s}|^s} \right]^\frac{1}{q} \\
 & \leq 2^\frac{1}{q} \mb E^{\mu \otimes Q}\left[  m(X_t,X_{t+1})^p \right]^\frac{1}{p} \,,
\end{align*}
where the final line uses definition of $\|\cdot\|_{\phi_s}$. Note all moments of $m$ are finite under condition (\ref{e:m-thin}). It follows by Lemma \ref{lem:pollard} and definition of the operator norm $\| \mb D \|_{L^{\phi_s}}$ that
\[
 \| \mb D \|_{L^{\phi_s}} \leq \left( \left( 2^\frac{1}{q} \mb E^{\mu \otimes Q}\left[  m(X_t,X_{t+1})^p \right]^\frac{1}{p} - 1 \right) \vee 1 \right) q^\frac{1}{s}\beta < \infty \,.
\]
That $\mb D : E^{\phi_s} \to E^{\phi_s}$ may be deduced similarly. Boundedness of $\mb D$ on $E^{\phi_s}$ now follows because $E^{\phi_s}$ is a closed linear subspace of $L^{\phi_s}$.

We use Lemma \ref{lem:msubexp} to establish the spectral radius condition. We prove the result for the spaces $L^{\phi_s}$; the results for $E^{\phi_s}$ follow because $E^{\phi_s}$ is a closed linear subspace of $L^{\phi_s}$.
First consider the case with $s > 1$. Fix $p,q \in (1,\infty)$ with $p^{-1} + q^{-1} = 1$. For any $f \in L^{\phi_s}$ with $\|f\|_{\phi_s} > 0$, by two applications of Jensen's inequality we have
\begin{align*}
 \mb E^{\mu}\left[ e^{|\mb D^n f(X_t)/(q^\frac{1}{s}(\beta^{\frac{s-1}{s}})^n \|f\|_{\phi_s})|^s} \right] 
 & = \mb E^{\mu}\left[ e^{\beta^n q^{-1} |\tilde{\mb E}^n f(X_t)/\|f\|_{\phi_s}|^s} \right] \\
 & \leq \mb E^{\mu}\left[ e^{q^{-1} |\tilde{\mb E}^n f(X_t)/\|f\|_{\phi_s}|^s} \right]^{ \beta^n}  \leq \mb E^{\mu}\left[ \tilde{\mb E}^n g(X_t) \right]^{ \beta^n} \,,
\end{align*}
where $g(x) = \exp( q^{-1} |f(x)/\|f\|_{\phi_s}|^s)$. By H\"older's inequality, 
\begin{align*}
 \mb E^{\mu}\left[ \tilde{\mb E}^n g(X_t) \right] & = \mb E^{\mu \otimes Q}\left[ m(X_t,X_{t+1}) \cdots m(X_{t+n-1},X_{t+n}) g(X_{t+n}) \right] \\
 & \leq \mb E^{\mu \otimes Q}\left[ \left( m(X_t,X_{t+1}) \cdots m(X_{t+n-1},X_{t+n}) \right)^p \right]^{\frac{1}{p}} \mb E^\mu \left[ |g(X_t)|^q \right]^\frac{1}{q} \\
 & \leq \mb E^{\mu \otimes Q}\left[ m(X_t,X_{t+1})^{np} \right]^\frac{1}{np} \cdots \mb E^{\mu \otimes Q}\left[ m(X_{t+n-1},X_{t+n})^{np} \right]^\frac{1}{np} \mb E^\mu \left[ |g(X_t)|^q \right]^\frac{1}{q} \\
 & = \mb E^{\mu \otimes Q}\left[ m(X_t,X_{t+1})^{np} \right]^\frac{1}{p} \mb E^\mu \left[ |g(X_t)|^q \right]^\frac{1}{q} \,.
\end{align*}
It follows by Lemma \ref{lem:msubexp}, and definition of $g$ and $\|\cdot\|_{\phi_s}$ that
\begin{align*}
 \mb E^{\mu}\left[ \tilde{\mb E}^n g(X_t) \right] 
 \leq & \mb E^{\mu \otimes Q} \left[ m(X_t,X_{t+1})^{np} \right]^\frac{1}{p} \mb E^\mu \left[ e^{|f(X_t)/\|f\|_{\phi_s}|^s} \right]^\frac{1}{q}  \leq  2^\frac{1}{q} C e^{(\beta+c)^{-n}} 
\end{align*}
for constants  $C \in (0,\infty)$ and $c \in (0,1-\beta)$ not depending on $f$. Therefore,
\[
 \mb E^{\mu}\left[ e^{|\mb D^n f(X_t)/(q^\frac{1}{s}(\beta^{\frac{s-1}{s}})^n \|f\|_{\phi_s})|^s} \right] 
 \leq \left(2^\frac{1}{q} C e^{(\beta+c)^{-n}} \right)^{\beta^n} \,.
\]
It follows by Lemma \ref{lem:pollard} and definition of the operator norm $\| \mb D^n \|_{L^{\phi_s}}$ that
\[
 \| \mb D^n \|_{L^{\phi_s}} \leq \left( \left( \left(2^\frac{1}{q} C e^{(\beta+c)^{-n}} \right)^{\beta^n} - 1 \right) \vee 1 \right) q^\frac{1}{s}(\beta^\frac{s-1}{s})^n 
\]
and therefore $\rho(\mb D;L^{\phi_s}) \equiv \lim_{n \to \infty} \| \mb D^n \|_{L^{\phi_s}}^{1/n} \leq \beta^\frac{s-1}{s} < 1$.

Now consider the case with $s = 1$. Let $c$ be as in Lemma \ref{lem:msubexp}. Fix any $\varepsilon \in (0,1)$ and note that $\beta  < \beta + \varepsilon c < \beta + c < 1$. For any $f \in L^{\phi_1}$ with $\|f\|_{\phi_1} > 0$, we have:
\begin{align*}
 \mb E^{\mu}\left[ e^{|\mb D^n f(X_t)/(q \beta^n(\beta+\varepsilon c)^{-n} \|f\|_{\phi_1})|} \right] 
 & = \mb E^{\mu}\left[ e^{(\beta+\varepsilon c)^n q^{-1} |\tilde{\mb E}^n f(X_t)/\|f\|_{\phi_1}|} \right] \\
 & \leq \mb E^{\mu}\left[ e^{q^{-1} |\tilde{\mb E}^n f(X_t)/\|f\|_{\phi_1}|} \right]^{(\beta+\varepsilon c)^n} \\
 & \leq \mb E^{\mu}\left[ \tilde{\mb E}^n g(X_t) \right]^{(\beta+\varepsilon c)^n} \,,
\end{align*}
where $g(x) = \exp(q^{-1}|f(x)|/\|f\|_{\phi_1})$. By similar arguments to above, we obtain
\[
 \mb E^{\mu}\left[ e^{|\mb D^n f(X_t)/(q \beta^n(\beta+\varepsilon c)^{-n} \|f\|_{\phi_1})|} \right] 
 \leq (2^\frac{1}{q}  C e^{(\beta+c)^{-n}} )^{ (\beta+\varepsilon c)^n} \,.
\]
By Lemma \ref{lem:pollard} and definition of the operator norm $\| \mb D^n \|_{L^{\phi_1}}$, we may deduce that
\[
 \| \mb D^n \|_{\phi_1} \leq \left( \left(  (2^\frac{1}{q}  C e^{(\beta+c)^{-n}} )^{ (\beta+\varepsilon c)^n} - 1 \right) \vee 1 \right) q \left( \frac{\beta}{\beta+\varepsilon c} \right)^n  \,,
\]
from which it follows similarly that $\rho(\mb D;L^{\phi_1}) \equiv \lim_{n \to \infty} \| \mb D^n \|_{L^{\phi_1}}^{1/n} \leq \frac{\beta}{\beta+\varepsilon c} < 1$.
\end{proof}

\subsection{Proofs for Section \ref{s:robust}}

\begin{proof}[Proof of Theorem \ref{t-id-W}]
We verify the conditions of Proposition \ref{p:exun}. 
For existence, Lemma \ref{lem:T-prop} shows $\mb T$ is a continuous, monotone, and convex operator on $E^{\phi_s}$ for each $1 \leq s \leq r$. Let
\[
 \bar v(x) = (1-\beta) \sum_{n=0}^\infty \beta^{n+1} \log \left((\mb E^{Q})^n h(x)\right) \,,
\]
where $h(x)= \mb E^Q[ e^{\frac{\alpha}{1-\beta} u(X_{t},X_{t+1})} |X_t = x ]$. We first show that $\mb E^{\mu}[\exp(|\bar v(X_t)/(\beta c)|^r)] < \infty$ holds for each $c \in (0,1]$. By Jensen's inequality (using the fact that $\sum_{n=1}^\infty (1-\beta) \beta^n = 1$ and convexity of $x \mapsto e^{|x/c|^r}$ and $x \mapsto e^{|(\log x)/c|^r}$ for $c \in (0,1]$), we obtain
\begin{align*}
 \mb E^{\mu}\left[ e^{\left|\bar v(X_t)/(\beta c)\right|^r} \right] 
 & = \mb E^{\mu}\left[ \exp \left( \left| (1-\beta) \sum_{n=0}^\infty \beta^n \log \left((\mb E^{Q})^n h(X_t)\right)/c\right|^r \right) \right] \\
 & \leq (1-\beta) \sum_{n=0}^\infty \beta^{n}   \mb E^{\mu}\left[  \exp \left( \left| \log \left((\mb E^{Q})^n h(x)\right)/c \right|^r \right) \right] \\
 & \leq (1-\beta) \sum_{n=0}^\infty \beta^n \mb E^{\mu \otimes Q}\Big[ e^{|\frac{\alpha}{c(1-\beta)} u(X_{t+n},X_{t+n+1})|^r}\Big] \\
 & = \mb E^{\mu \otimes Q}\Big[ e^{|\frac{\alpha}{c(1-\beta)} u(X_{t+n},X_{t+n+1})|^r}\Big] <\infty \,.
\end{align*}
It follows by Remark \ref{rmk:r3} that $\bar v \in E^{\phi_r}$.

We now show that $\mb T \bar v \leq \bar v$. By Holder's inequality we first have
\begin{align}
 \mb T \bar v(X_t) & \leq \beta \log \left( \mb E^Q \Big[ e^{\bar v(X_{t+1})/\beta} \Big|X_t \Big]^{\beta} \mb E^Q \Big[ e^{\frac{\alpha}{1-\beta} u(X_t,X_{t+1})} \Big|X_t \Big]^{1-\beta} \right) \notag \\
 & = \beta^2  \log  \mb E^Q [ e^{\bar v(X_{t+1})/\beta}|X_t ] + (1-\beta) \beta \log h(X_t) \,. \label{e-exist-1}
\end{align} 
By Lemma \ref{lem:holder-infinite}, we may deduce
\begin{align}
 \log  \mb E^Q \left[ \left. e^{\bar v(X_{t+1})/\beta} \right|X_t \right] 
 & = \log  \mb E^Q \left[ \left.  \prod_{n=0}^\infty  \left((\mb E^{Q})^n h(X_{t+1})\right)^{(1-\beta)  \beta^{n}}  \right| X_t\right] \notag \\
 & \leq \log  \left( \prod_{n=0}^\infty  \mb E^Q \left[ \left.    \left((\mb E^{Q})^n h(X_{t+1})\right)  \right| X_t\right]^{(1-\beta)  \beta^{n}} \right) \notag \\
 & = (1-\beta) \sum_{n=1}^\infty \beta^{n-1}  \log \left((\mb E^{Q})^n h(X_t)\right) \,. \label{e-exist-2}
\end{align}
Substituting (\ref{e-exist-2}) into (\ref{e-exist-1}) yields $\mb T \bar v \leq \bar v$. 

We now show $\{\mb T^n \bar v\}_{n \geq 1}$ is bounded from below, first observe that
\[
 \mb T f (x) = \beta \log \mb E^Q[ e^{f(X_{t+1}) + \alpha u(X_t,X_{t+1})} | X_t = x] \geq \beta \mb E^Q[ f(X_{t+1}) + \alpha u(X_t,X_{t+1}) | X_t = x] \,.
\]
Therefore, 
\[
 \mb T^n \bar v \geq (\beta \mb E^Q)^n \bar v + \sum_{s=0}^{n-1}  (\beta \mb E^Q)^s ( h_1) 
\]
for each $n \geq 1$, where $h_1(x) = \beta \mb E^Q[\alpha u(X_t,X_{t+1})|X_t =x]$. 
Note also that $\|\beta \mb E^Q\|_{E^{\phi_r}} = \beta$ and $\rho(\beta \mb E^Q;E^{\phi_r}) = \beta$ (see Section \ref{s:sr}), and so we obtain $\liminf_{n \to \infty} \mb T^n \bar v \geq (\mb I - \beta \mb E^Q)^{-1} h_1 \in E^{\phi_r}$.

Uniqueness: $v$ is a fixed point of $\mb T : E^{\phi_s} \to E^{\phi_s}$ for each $s \in [1,r]$. Moreover, $\mb T : E^{\phi_s} \to E^{\phi_s}$ is convex by Lemma \ref{lem:T-prop} and $\mb D_v$ is a bounded, monotone linear operator with $\rho(\mb D_v;E^{\phi_s}) < 1$ for $s \in [1,r]$ by Lemma \ref{lem:D-bdd}. Uniqueness in $E^{\phi_s}$ with $s \in (1,r]$ follows by Proposition \ref{p:exun}(ii). That $v$ is the smallest and unique stable fixed point in $E^{\phi_1}$ follows  by Corollary \ref{cor:nonunique-gen}.
\end{proof}

\begin{lemma}\label{lem:T-prop}
Let condition (\ref{e:u}) hold. Then: $\mb T$ is a continuous, monotone and convex operator on $E^{\phi_s}$ for each $1 \leq s \leq r$.
\end{lemma}

\begin{proof}[Proof of Lemma \ref{lem:T-prop}]
Fix any $1 \leq s \leq r$. Take any $f \in E^{\phi_s}$ and $c \in (0,1]$. By convexity of $x \mapsto e^{|(\log x)/c|^s}$ for $c \in (0, 1]$ and Jensen's inequality:
\begin{align*}
 \mb E^{\mu}[\exp(|\mb T f(X_t)/(\beta c)|^s)]
 & = \mb E^{\mu} \left[\exp \left(\left|\frac{1}{c}\log \mb E^Q \left[ \left. e^{f(X_{t+1}) + \alpha u(X_t,X_{t+1})}\right|X_t\right]\right|^s\right)\right] \\ 
 & \leq \mb E^{\mu} \left[ \mb E^Q \left[ \left. \exp \left(\left|\frac{1}{c}\log  e^{f(X_{t+1}) + \alpha u(X_t,X_{t+1})}\right|^s\right) \right|X_t\right] \right] \\ 
 & = \mb E^{\mu \otimes Q} \left[  \exp \left(\left|\frac{f(X_{t+1}) + \alpha u(X_t,X_{t+1})}{c}\right|^s\right) \right] < \infty
\end{align*}
which is finite for any $f \in E^{\phi_s}$ under (\ref{e:u}). It follows by Remark \ref{rmk:r3} that $\mb T: E^{\phi_s} \to E^{\phi_s}$.

Continuity: Fix any $f\in E^{\phi_s}$. Take $g \in E^{\phi_s}$ with $\|g\|_{\phi_s}  \in (0,2^{-1/s}]$ and set $c = 2^{1/s} \|g\|_{\phi_s}$. Let $\mb E_f$ denote the distorted conditional expectation operator from (\ref{eq:mf-robust}) with $f$ in place of $v$. By convexity of $x \mapsto e^{|(\log x)/c|^s}$ for $c \in (0, 1]$ and the Jensen and Cauchy-Schwarz inequalities,
\begin{align*}
 \mb E^{\mu} \left[ \phi_s(| \mb T(f + g)(X_t) - \mb T f(X_t)|/(\beta c)) \right] +1
 & = \mb E^{\mu} \left[\exp \left(\left|\frac{1}{c}\log \mb E_f \left[ \left. e^{g(X_{t+1}) }\right|X_t\right]\right|^s\right)\right] \\
 & \leq \mb E^{\mu} \left[\mb E_f \left[ \left. \exp \left(\left|\frac{1}{c}\log  e^{g(X_{t+1}) } \right|^s\right) \right|X_t\right] \right] \\
 & = \mb E^{\mu \otimes Q} \left[m_f(X_t,X_{t+1})  \exp \left(\left|\frac{g(X_{t+1})}{c} \right|^s\right)  \right] \\
 & \leq \mb E^{\mu} \big[ e^{2|g(X_t)/c|^s} \big]^{1/2} \mb E^{\mu \otimes Q}[m_f(X_t,X_{t+1})^2]^{1/2} \\
 & =  \sqrt {2 \mb E^{\mu \otimes Q}[m_f(X_t,X_{t+1})^2] }
\end{align*}
because $c = 2^{1/s} \|g\|_{\phi_s}$. Finiteness of $\mb E^{\mu \otimes Q}[m_f(X_t,X_{t+1})^2]$ holds for any $f \in E^{\phi_s}$ under (\ref{e:u}). To see this, by several applications of the Cauchy--Schwarz and Jensen inequalities, we have
\begin{align*}
 \mb E^{\mu \otimes Q}[m_f(X_t,X_{t+1})^2] & = \mb E^{\mu \otimes Q}\left[\left(\frac{e^{f(X_{t+1}) + \alpha u(X_t,X_{t+1})}}{\mb E^Q[e^{f(X_{t+1}) + \alpha u(X_t,X_{t+1})}|X_t]}\right)^2\right]  \\
 & \leq \mb E^{\mu \otimes Q}\left[ e^{4|f(X_{t+1}) + \alpha u(X_t,X_{t+1})|} \right] \\
 & \leq \mb E^{\mu }\left[ e^{8|f(X_t) |} \right]^{1/2} \mb E^{\mu \otimes Q}\left[ e^{8|\alpha u(X_t,X_{t+1})|} \right]^{1/2} \,,
\end{align*}
which is finite for any $f \in E^{\phi_s}$ under (\ref{e:u}). Continuity now follows by Lemma \ref{lem:pollard}.
Monotonicity follows from monotonicity of $\exp(\cdot)$, $\log(\cdot)$, and conditional expectations. Convexity follows by applying H\"older's inequality to the conditional expectation 
\[
 \mb E^Q \left[ \left. e^{\tau(v_1(X_{t+1}) + \alpha u(X_t,X_{t+1})) + (1-\tau)(v_2(X_{t+1}) + \alpha u(X_t,X_{t+1}))}  \right|X_t=x \right]
\]
with $p = \tau^{-1}$ and $q = (1-\tau)^{-1}$.
\end{proof}

\begin{lemma}\label{lem:D-bdd}
Let condition (\ref{e:u}) hold with $r > 1$ and fix any $v \in E^{\phi_{r'}}$ with $r' > 1$. Then: for each $s \geq 1$, $\mb D_v$ is a continuous linear operator on $E^{\phi_s}$ with $\rho(\mb D_v;E^{\phi_s}) < 1$.
\end{lemma}

\begin{proof}[Proof of Lemma \ref{lem:D-bdd}]
We verify condition (\ref{e:m-thin}) from Lemma \ref{lem:com}. The log change-of-measure is
\[
 \log m_v(X_t,X_{t+1}) = v(X_{t+1}) + \alpha u(X_t,X_{t+1}) - \log \mb E^Q[ e^{v(X_{t+1}) + \alpha u(X_t,X_{t+1})}|X_t] \,.
\]
For any $v \in E^{\phi_{r'}}$ with $r' > 1$, setting $\ul r = (r \wedge r') > 1$ and taking any $c \in (0,1]$, 
\[
 \mb E^{\mu} \left[ e^{|\log \mb E^Q[ e^{v(X_{t+1}) + \alpha u(X_t,X_{t+1})}|X_t]/c|^{\ul r}} \right] \leq \mb E^{\mu \otimes Q} \left[ e^{|(v(X_{t+1}) + \alpha u(X_t,X_{t+1}))/c|^{\ul r}} \right]
\]
by Jensen's inequality. The right-hand side is finite by condition (\ref{e:u}). Therefore, 
\[
 \mb E^{\mu \otimes Q} \left[ e^{|\log m_v(X_t,X_{t+1})/c|^{\ul r}} \right] < \infty
\]
for any $c \in (0,1]$ and hence for any $c > 0$ (see Remark \ref{rmk:r3}), verifying (\ref{e:m-thin}). 
\end{proof}

\begin{lemma}\label{lem:hn}
Let $Y = |Z|$ with $Z \sim N(0,1)$. Then for $a > 0$ and $r \in [1,2)$, we have
\[
 \mathbb E\left[\exp\left(\frac{Y^r}{a^r}\right)\right] \leq \frac{\sqrt 2}{\sqrt \pi} \left( \left(\frac{2}{a^r}\right)^\frac{1}{2-r} \exp \left( \frac{2^\frac{r}{2-r}}{a^\frac{2r}{2-r}} \right) + \left(\frac{4}{a^r}\right)^\frac{1}{2-r} + \sqrt \pi \right) \,.
\]
\end{lemma}

\begin{proof}[Proof of Lemma \ref{lem:hn}]
First write
\begin{align*}
 \mathbb E\left[\exp\left(\frac{Y^r}{a^r}\right)\right] & = \frac{\sqrt 2}{\sqrt \pi} \int_0^\infty \exp \left( \frac{y^r}{a^r} - \frac{1}{2}y^2 \right) \mathrm d y \\
 & \leq \frac{\sqrt 2}{\sqrt \pi} \left( \int_0^{(\frac{2}{a^r})^\frac{1}{2-r}} \exp \left( \frac{y^r}{a^r} \right) \mathrm d y + \int_{(\frac{2}{a^r})^\frac{1}{2-r}}^{(\frac{4}{a^r})^\frac{1}{2-r}} \mathrm d y + \int_{(\frac{4}{a^r})^\frac{1}{2-r}}^\infty \exp \left( - \frac{1}{4}y^2 \right) \mathrm d y\right) \\
 & \leq \frac{\sqrt 2}{\sqrt \pi} \left( \left(\frac{2}{a^r}\right)^\frac{1}{2-r} \exp \left( \frac{2^\frac{r}{2-r}}{a^\frac{2r}{2-r}} \right) + \left(\frac{4}{a^r}\right)^\frac{1}{2-r} +  \sqrt \pi \right) \,.
\end{align*}
The first inequality follows by noting that $\frac{y^r}{a^r} - \frac{1}{2}y^2 \leq \frac{y^r}{a^r}$ (for the first integral), $\frac{y^r}{a^r} - \frac{1}{2}y^2 \leq 0$ over $[(\frac{2}{a^r})^\frac{1}{2-r},\infty)$ (for the second integral), and $\frac{y^r}{a^r} - \frac{1}{2}y^2 \leq - \frac{1}{4}y^2$ over $[(\frac{4}{a^r})^\frac{1}{2-r},\infty)$ (for the third integral). For the three integrals on the second line, the first is bounded using the inequality $\int_0^b \exp ( \frac{y^r}{a^r} ) d y \leq b \exp ( \frac{b^r}{a^r} )$ (valid for $b \geq 0$); the second and third are trivial.
\end{proof}

\begin{proof}[Proof of Proposition \ref{prop:cpt-approx}]
The boundedness condition in the statement of the lemma ensures $\mb T_{\mc C}$ is a self map on $B(\mc C)$. It is then straightforward to verify that $\mb T_{\mc C}$ satisfies Blackwell's sufficient conditions, and therefore has a unique fixed point $v_{\mc C} \in B(\mc C)$.

To relate $v$ and $v_{\mc C}$, let $v_|$ denote the restriction of $v$ to $\mc C$. Then for $x \in \mc C$, we have
\begin{align*}
 v(x) - v_{\mc C}(x) & = \beta \log \mb E^Q[e^{v(X_{t+1}) + \alpha u(X_t,X_{t+1})}|X_t = x] - \mb T_{\mc C} v_{\mc C}(x) \\
 & \geq \beta \log \mb E^Q[ e^{v_|(X_{t+1}) + \alpha u(X_t,X_{t+1})} \ind \{X_{t+1} \in \mc C \} | X_t = x ] - \mb T_{\mc C} v_{\mc C}(x) \\
 & = \beta \log Q(\mc C | x ) + \mb T_{\mc C} v_{|}(x) - \mb T_{\mc C} v_{\mc C}(x) \\
 & \geq \beta \log Q(\mc C | x ) + \beta \mb E^Q \left[ \left. m_{\mc C,v_{\mc C}}(X_t,X_{t+1}) (v_|(X_{t+1}) - v_{\mc C}(X_{t+1}))  \right| X_t = x \right] \\
 & \geq \beta \log Q(\mc C | x ) + \beta \inf_{x \in \mc C} \left( v(x) - v_{\mc C}(x) \right) \,,
\end{align*}
where the first inequality is by monotonicity of expectations, the second equality is because $\inf_{x \in \mc C} Q(\mc C | x) > 0$, and the second inequality is by Jensen's inequality with
\[
 m_{\mc C,v_{\mc C}}(X_t, X_{t+1}) = \frac{e^{v_{\mc C}(X_{t+1}) + \alpha u(X_t,X_{t+1})}\ind\{X_{t+1} \in \mc C\}}{\mb E^Q [ e^{v_{\mc C}(X_{t+1}) + \alpha u(X_t,X_{t+1})} \ind\{X_{t+1} \in \mc C\}| X_t ]} \,.
\]
The result follows by taking the infimum of both sides with respect to $x \in \mc C$.
\end{proof}

\subsection{Proofs for Section \ref{s:learn}} \label{s:proofs-learn}

Recall $\hat X_t = (\hat \xi_t, \varphi_t)$. The conditional distribution $\hat Q$ of $(\xi_t,\hat X_{t+1})$ given $\hat X_t$ may be represented by 
\[
 \mb E^{\hat Q} [h(\xi_t,\hat X_{t+1})|\hat X_t] = \mb E^{\hat Q} [h(\xi_t,\hat X_{t+1})|\hat \xi_t] = \mb E^{\Pi_\xi \otimes Q_\varphi}[ h(\xi_t,\varphi_{t+1},\Xi(\hat \xi_t,\varphi_{t+1})) | \hat \xi_t] \,.
\]
Recall that $\mu$ is the stationary distribution of $\hat X_t$ under $\hat Q$. For $v \in E^{\phi_1}_{\hat \xi}$, define
\begin{align*}
 m_v^{\Pi_\xi}(\xi_t,\hat \xi_t) & = \frac{ \mb E^{Q_\varphi} \left[ \left.  e^{\frac{\theta}{\vartheta} v(\Xi(\hat \xi_t,\varphi_{t+1})) + \alpha u(\varphi_{t+1})} \right| \xi_t ,\hat \xi_t \right]^\frac{\vartheta}{\theta} }{\mb E^{\Pi_\xi}\! \left[ \left. \mb E^{Q_\varphi} \left[ \left.  e^{\frac{\theta}{\vartheta} v(\Xi(\hat \xi_t,\varphi_{t+1})) + \alpha u(\varphi_{t+1})} \right| \xi_t ,\hat \xi_t \right]^\frac{\vartheta}{\theta}  \right| \hat \xi_t \right]} \\
 m_v^{Q_\varphi}(\xi_t,\hat \xi_t,\varphi_{t+1}) & = \frac{e^{\frac{\theta}{\vartheta} v(\Xi(\hat \xi_t,\varphi_{t+1})) + \alpha u(\varphi_{t+1})}}{\mb E^{Q_\varphi} \left[ \left.  e^{\frac{\theta}{\vartheta} v(\Xi(\hat \xi_t,\varphi_{t+1})) + \alpha u(\varphi_{t+1})} \right| \xi_t, \hat \xi_t \right]} \,.
\end{align*}
The quantity $m_v^{\Pi_\xi}$ distorts the posterior distribution for $\xi_t$ given $\hat X_t$ whereas $m_v^{ Q_\varphi}$ distorts the conditional distribution $Q_\varphi$. To simplify notation, define the distorted conditional expectations $\mb E^{\Pi_\xi}_v$ and  $\mb E^{Q_\varphi}_v$ by
\begin{align*}
 \mb E^{\Pi_\xi}_v f(\hat \xi) & = \mb E^{\Pi_\xi} \left[ \left. m_v^{\Pi_\xi}(\xi_t, \hat \xi_t) f(\xi_t,\hat \xi_t) \right| \hat \xi_t = \hat \xi \right] \,, \\
 \mb E^{Q_\varphi}_v f(\xi, \hat \xi) & = \mb E^{\Pi_\xi} \left[ \left. m_v^{Q_\varphi}(\xi_t, \hat \xi_t, \varphi_{t+1}) f(\xi_t,\hat \xi_t, \varphi_{t+1}) \right| \xi_t = \xi, \hat \xi_t = \hat \xi \right] \,.
\end{align*}
The subgradient of $\mb T$ at $v$ is the composition of these two distorted conditional expectations, discounted by $\beta$:
\begin{equation} \label{e:subgradient-learning}
 \mb D_v f(\hat \xi) = \beta \mb E^{\hat Q} \left[ \left. m_v (\xi_t,\hat \xi_t,\varphi_{t+1}) f(\hat \xi_{t+1}) \right| \hat \xi_t=\hat \xi \right] 
\end{equation}
where $m_v (\xi_t,\hat \xi_t,\varphi_{t+1}) = m_v^{\Pi_\xi}(\xi_t,\hat \xi_t) m_v^{Q_\varphi}(\xi_t,\hat \xi_t,\varphi_{t+1})$.

\begin{proof}[Proof of Theorem \ref{t-id-W-learn}]
We verify the conditions of Proposition \ref{p:exun}. Lemma \ref{lem:T-prop-learn} shows that $\mb T$ is a continuous, monotone, and convex operator on $E^{\phi_s}_{\hat \xi}$ for each $1 \leq s \leq r$. If $\theta < \vartheta$, let
\[
 \bar v(\hat \xi) = (1-\beta) \sum_{n=0}^\infty \beta^{n+1} \log \left( \left(\mb E^{\hat Q}\right)^{n+1} g_1(\hat \xi) \right) \,, 
\]
where $g_1(\hat X_t) =  \exp(\frac{\alpha\vartheta}{(1-\beta)\theta} u(\varphi_t))$. For any $c > 0$, by Jensen's inequality we may deduce
\begin{align*}
 \mb E^{\mu}[ e^{|\bar v(\hat \xi_t)/(\beta c)|^r}] & \leq (1-\beta) \sum_{n=0}^\infty \beta^n \mb E^{\mu}  \left[\left( \left(\mb E^{\hat Q}\right)^{n+1} g_1^r (\hat \xi_t) \right) \right] \,,
\end{align*}
where $g_1^r(\hat X_t) = \exp(|\frac{\alpha\vartheta}{(1-\beta)\theta c} u(\varphi_t)|^r)$. As $u \in E^{\phi_r}_\varphi$, the right-hand side of the preceding display is finite and so $\bar v \in E^{\phi_r}_{\hat \xi}$. 

To show $\mb T \bar v \leq \bar v$, first by the Jensen and H\"older inequalities,
\begin{align*}
 \mb T \bar v(\hat \xi) & = \beta \log \mb E^{\Pi_\xi}\! \left[ \left. \mb E^{Q_\varphi} \left[ \left.  e^{\frac{\theta}{\vartheta} \bar v(\Xi(\hat \xi_t,\varphi_{t+1})) + \alpha u(\varphi_{t+1})} \right| \xi_t, \hat \xi_t \right]^{\vartheta/\theta}  \right| \hat \xi_t = \hat \xi \right] \\
 & \leq \beta \log \mb E^{\hat Q} \left[ \left.  e^{ \bar v(\hat \xi_{t+1}) + \alpha \frac{\vartheta}{\theta} u(\varphi_{t+1})}  \right| \hat \xi_t = \hat \xi \right] \\ 
 & \leq \beta^2 \log \mb E^{\hat Q} \left[ \left.  e^{ \bar v(\hat \xi_{t+1})/\beta }  \right| \hat \xi_t = \hat \xi \right] +  \beta(1-\beta) \log \mb E^{\hat Q} \left[ \left.  e^{  \frac{\alpha\vartheta}{(1-\beta)\theta} u(\varphi_{t+1})}  \right| \hat \xi_t = \hat \xi \right]  \,.
\end{align*}
By Lemma \ref{lem:holder-infinite}, we may deduce
\[
 \log \mb E^{\hat Q} \left[ \left.  e^{ \bar v(\hat \xi_{t+1})/\beta }  \right| \hat \xi_t = \hat \xi \right] \leq (1-\beta) \sum_{n=1}^\infty \beta^{n-1} \log \left( \Big(\mb E^{\hat Q}\Big)^{n+1} g_1(\hat \xi) \right)  \,,
\]
hence $\mb T \bar v \leq \bar v$. 

On the other hand, if $\vartheta \leq \theta$, let $\bar v(\hat \xi) = \frac{\vartheta}{\theta} (1-\beta) \sum_{n=0}^\infty \beta^{n+1} \log ( (\mb E^{\hat Q})^{n+1} g_2(\hat \xi) )$ where $g_2(\hat X_t) = e^{\frac{\alpha}{1-\beta} u(\varphi_t)}$. By similar arguments to above, we may use the condition $u \in E^{\phi_r}_\varphi$ to deduce $\bar v \in E^{\phi_r}_{\hat \xi}$. Again  by the Jensen and H\"older inequalities,
\begin{align*}
 \mb T \bar v(\hat \xi) & = \beta \log \mb E^{\Pi_\xi}\! \left[ \left. \mb E^{Q_\varphi} \left[ \left.  e^{\frac{\theta}{\vartheta} \bar v(\Xi(\hat \xi_t,\varphi_{t+1})) + \alpha u(\varphi_{t+1})} \right| \hat \xi_t, \xi_t \right]^\frac{\vartheta}{\theta}  \right| \hat \xi_t = \hat \xi \right] \\
 & \leq \frac{\vartheta}{\theta} \beta \log \mb E^{\hat Q} \left[ \left. e^{\frac{\theta}{\vartheta} \bar v(\hat \xi_{t+1}) + \alpha u(\varphi_{t+1})}  \right| \hat \xi_t = \hat \xi \right] \\ 
 & \leq \frac{\vartheta}{\theta}  \beta^2 \log \mb E^{\hat Q} \left[ \left.  e^{ \frac{\theta}{\vartheta} \bar v(\hat \xi_{t+1})/\beta }  \right| \hat \xi_t = \hat \xi \right] + \frac{\vartheta}{\theta}  \beta(1-\beta) \log \mb E^{\hat Q} \left[ \left.  e^{  \frac{\alpha}{1-\beta} u(\varphi_{t+1})}  \right| \hat \xi_t = \hat \xi \right]  \,.
\end{align*}
The inequality $\mb T \bar v \leq \bar v$ now follows by similar arguments to the previous case. 

To show that the sequence of iterates $\mb T^n \bar v$ is bounded from below, first note that for any $f \in E^{\phi_r}_{\hat \xi}$, we have
\[
 \mb T f(\hat \xi) \geq \beta \mb E^{\hat Q}\left[ \left.  f(\hat \xi_{t+1}) + \alpha \frac{\vartheta}{\theta} u(\varphi_{t+1})   \right| \hat \xi_t = \hat \xi \right]
\]
which follows by several applications of Jensen's inequality. It follows that
\[
 \mb T^n \bar v(\hat \xi) \geq \left(\beta \mb E^{\hat Q} \right)^n \bar v (\hat \xi) + \sum_{i=0}^{n-1} \left(\beta \mb E^{\hat Q}\right)^i g_3(\hat \xi)
\]
where $g_3(\hat \xi) = \beta \mb E^{\hat Q}[ \alpha \frac{\vartheta}{\theta} u(\varphi_{t+1}) | \hat \xi_t = \hat \xi ] \in E^{\phi_r}_{\hat \xi}$. Note also that $\rho(\beta \mb E^{\hat Q};E^{\phi_r}) = \beta$ (see Section \ref{s:sr}), hence $\liminf_{n \to \infty} \mb T^n \bar v \geq (\mb I - \beta \mb E^{\hat Q})^{-1} g_3 \in E^{\phi_r}$. This completes the proof of existence.

For uniqueness, $v$ is necessarily a fixed point of $\mb T : E^{\phi_s}_{\hat \xi} \to E^{\phi_s}_{\hat \xi}$ for each $1 \leq s \leq r$. The subgradient $\mb D_v$ is  monotone. Lemma \ref{lem:D-bdd-learn} shows $\mb D_v : E^{\phi_s}_{\hat \xi} \to E^{\phi_s}_{\hat \xi}$ is bounded and $\rho(\mb D_v;E^{\phi_s}_{\hat \xi}) < 1$ for $s \in [1,r]$.  Uniqueness follows by Proposition \ref{p:exun}(ii) and Corollary \ref{cor:nonunique-gen}.
\end{proof}

\begin{lemma} \label{lem:T-prop-learn}
Let condition (\ref{e:u-learn}) hold. Then: $\mb T$ is a continuous, monotone, and convex operator on $ E^{\phi_s}_{\hat \xi}$ for each $1 \leq s \leq r$.
\end{lemma}

\begin{proof}[Proof of Lemma \ref{lem:T-prop-learn}]
Fix $s \in [1,r]$. We first show $\mb E^{\mu}[\exp(|\mb T f(\hat \xi_t)/(\beta c)|^s)] < \infty$ holds for each $f \in E^{\phi_s}_{\hat \xi}$ and $c \in (0,\frac{\vartheta}{\theta} \wedge 1]$. By convexity of $x \mapsto e^{|(\log x)/c|^s}$ for $c \in (0,1]$ and Jensen's inequality, 
\begin{align*}
 \mb E^{\mu}\left[\exp\left(\left|\frac{\mb T f(\hat \xi_t)}{\beta c}\right|^s\right)\right] 
 & = \mb E^{\mu} \left[ \exp \left( \frac{1}{c^s} \left| \log \mb E^{\Pi_\xi}\! \left[ \left. \mb E^{Q_\varphi} \left[ \left. e^{\frac{\theta}{\vartheta} f(\Xi(\hat \xi_t,\varphi_{t+1})) + \alpha u (\varphi_{t+1})} \right| \xi_t, \hat \xi_t \right]^\frac{\vartheta}{\theta}  \right| \hat \xi_t \right]  \right|^s \right) \right] \\
 & \leq \mb E^{\mu} \left[ \mb E^{\Pi_\xi}\! \left[ \left. \exp \left( \frac{1}{c^s}  \left| \log  \mb E^{Q_\varphi} \left[ \left. e^{\frac{\theta}{\vartheta} f(\Xi(\hat \xi_t,\varphi_{t+1})) + \alpha u (\varphi_{t+1})} \right| \xi_t, \hat \xi_t \right]^\frac{\vartheta}{\theta}  \right|^s \right) \right| \hat \xi_t \right]\right] \\ 
 & \leq \mb E^{\mu} \left[ \mb E^{\Pi_\xi}\! \left[ \left. \mb E^{Q_\varphi} \left[ \left.  \exp \left( \frac{1}{c^s}  \left| \frac{\vartheta}{\theta} \log  e^{\frac{\theta}{\vartheta} f(\Xi(\hat \xi_t,\varphi_{t+1})) + \alpha u (\varphi_{t+1})}  \right|^s \right)  \right| \xi_t, \hat \xi_t \right]  \right| \hat \xi_t \right]\right] \\ 
 & = \mb E^{\mu \otimes \Pi_\xi \otimes Q_\varphi} \left[ \exp \left( \frac{1}{c^s}  \left| f(\Xi(\hat \xi_t,\varphi_{t+1}))  +\frac{\vartheta}{c\theta}  \alpha u (\varphi_{t+1})  \right|^s \right)  \right]  
\end{align*}
which is finite because $f \in {E}^{\phi_s}_{\hat \xi}$ and $u \in E^{\phi_r}_\varphi$. It follows by Remark \ref{rmk:r3} that $\mb T :E^{\phi_s}_{\hat \xi} \to E^{\phi_s}_{\hat \xi}$. 

For continuity, fix $f\in E^{\phi_s}_{\hat \xi}$. Take $g \in E^{\phi_s}_{\hat \xi}$ with $0 < \|g\|_{\phi_s} \leq 2^{-1/s}(1 \wedge \frac{\vartheta}{\theta})$ and set $c = 2^{1/s} \|g\|_{\phi_s}$. Note
\[
 \mb T(f+g)(\hat \xi) - \mb T f(\hat \xi) = \beta \log \left( \mb E^{\Pi_{\xi}}_f \left[ \left. \mb E^{Q_\varphi}_f \left[ \left. e^{\frac{\theta}{\vartheta} g(\Xi(\hat \xi_t,\varphi_{t+1}))} \right| \xi_t, \hat \xi_t \right]^\frac{\vartheta}{\theta} \right| \hat \xi_t=\hat \xi \right] \right)\,. 
\]
By similar arguments to the above, we may deduce
\begin{align*}
 \mb E^{\mu} \left[ \exp \left( \left| \frac{\mb T(f + g)(\hat \xi_t) - \mb T f(\hat \xi_t)}{\beta c} \right|^s \right) \right] 
 & \leq \mb E^{\mu} \left[ \mb E^{\Pi_\xi}_f\! \left[ \left. \mb E^{Q_\varphi}_f \left[ \left.  \exp \left(  \left| \frac{1}{c} g(\Xi(\hat \xi_t,\varphi_{t+1}))  \right|^s \right)  \right| \xi_t, \hat \xi_t \right]  \right| \hat \xi_t \right]\right] \\ 
 & = \mb E^{\mu} \left[ \mb E^{\hat Q} \left[ \left. m_f(\xi_t,\hat \xi_t,\varphi_{t+1}) \exp \left(  \left| \frac{1}{c} g(\Xi(\hat \xi_t,\varphi_{t+1}))  \right|^s \right)   \right| \hat \xi_t \right]\right] \\
 & \leq \mb E^{\mu \otimes \hat Q}\! \left[ m_f(\xi_t, \hat \xi_t, \varphi_{t+1})^2 \right]^{1/2}  \mb E^{\mu} \left[ \exp(2 |g(\hat \xi_{t+1})/c|^s \right]^{1/2}  \\
 & \leq  \left( 2 \mb E^{\mu \otimes \hat Q}\! \left[ m_f (\xi_t,\hat \xi_t,\varphi_{t+1})^2 \right] \right)^{1/2}  \,,
\end{align*}
because $c = 2^{1/s} \|g\|_{\phi_s}$. The expectation on the right-hand side is finite  because $f \in E^{\phi_s}_{\hat \xi}$ and $u \in E^{\phi_r}_\varphi$. It follows by Lemma \ref{lem:pollard} that $\|\mb T(f + g) - \mb T f\|_{\phi_s} \to 0$ as $\| g \|_{\phi_s} \to 0$.

Finally, monotonicity follows from monotonicity of the exponential and logarithm functions and monotonicity of conditional expectations. Convexity follows by H\"older's inequality.
\end{proof}

\begin{lemma}\label{lem:D-bdd-learn}
Let condition (\ref{e:u-learn}) hold. Fix any $v \in E^{\phi_{r'}}_{\hat \xi}$ with $r' > 1$. Then: for each $s \geq 1$, $\mb D_v$ is a continuous linear operator on $E^{\phi_s}_{\hat \xi}$ with $\rho(\mb D_v;E^{\phi_s}_{\hat \xi})  < 1$.
\end{lemma}

\begin{proof}[Proof of Lemma \ref{lem:D-bdd-learn}]
It suffices to verify the conditions of Lemma \ref{lem:com}. Note that the process $\hat \xi = \{\hat \xi_t\}_{t \in T}$ is a stationary Markov process (this follows from our maintained assumptions that learning is in a steady state and the conventional hidden Markov structure on $X$). By iterated expectations, we may rewrite the subgradient from (\ref{e:subgradient-learning}) as
\[
 \mb D_v f(\hat \xi)  = \beta \mb E^{\hat Q} \left[ \left. \bar m_v (\hat \xi_t,\hat \xi_{t+1}) f(\hat \xi_{t+1}) \right| \hat \xi_t=\hat \xi \right] 
\]
where $\bar m_v (\hat \xi_t, \hat \xi_{t+1})$ denotes the conditional expectation of $m_v (\xi_t,\hat \xi_t,\varphi_{t+1})$ given $\hat \xi_t,\hat \xi_{t+1}$ under $\hat Q$. 
The thin-tail condition on $m_v$ then follows by similar arguments to the proof of Lemma \ref{lem:D-bdd} for any $v \in E^{\phi_{r'}}_{\hat \xi}$ with $r' > 1$.
\end{proof}

\subsection{Proof for Section \ref{s:ez}}

\begin{proof}[Proof of Theorem \ref{t:ez-exist-1}]
In view of the discussion preceding Theorem \ref{t:ez-exist-1} and Lemma \ref{lem:T-prop-ez}, it suffices to show that $\bar v \in \tilde E^{\phi_r}$ and that $\mb T \bar v \leq \bar v$. By (\ref{e:eval}), convexity of $x \mapsto e^{|(\log x)/c|^r}$ for $c \in (0,1]$, and two applications of Jensen's inequality, for any $c \in (0,1]$ we have
\begin{align*}
 \mb E^{\tilde \mu} \left[ e^{|\bar v(X_t)/c|^r} \right]
 & = \mb E^{\tilde \mu} \left[ e^{\left|\log \left(  (1-\beta)\sum_{n=0}^\infty (\beta \lambda^{\frac{1}{\kappa}})^n \tilde{\mb E}^n(\iota^{-\frac{1}{\kappa}})(X_t) \right) /c\right|^r} \right] \\
 & \leq (1-\beta \lambda^\frac{1}{\kappa})\sum_{n=0}^\infty (\beta \lambda^{\frac{1}{\kappa}})^n \mb E^{\tilde \mu} \left[ \tilde{\mb E}^n \exp \left( \left|\log \left(  (1-\beta) (1-\beta \lambda^\frac{1}{\kappa})^{-1} (\iota(X_t))^{-\frac{1}{\kappa}} \right) /c \right|^r \right) \right] \\
 & = \mb E^{\tilde \mu} \left[  \exp \left( \left|\log \left(  (1-\beta) (1-\beta \lambda^\frac{1}{\kappa})^{-1} (\iota(X_t))^{-\frac{1}{\kappa}} \right) /c \right|^r \right) \right]  < \infty 
\end{align*} 
by condition (\ref{e:ez-phi-1}), with the final equality because $\tilde \mu$ is the stationary distribution corresponding to $\tilde{\mb E}$. It follows by Remark \ref{rmk:r3} that $\bar v \in E^{\phi_r}$.

To see that $\mb T \bar v \leq \bar v$, first note by Jensen's inequality that $\mb E[Z^\kappa]^{1/\kappa} \leq \mb E[Z]$ holds when $\kappa < 0$ for any random variable $Z$ that is (strictly) positive with probability $1$. Therefore,
\begin{align*}
 \mb T \bar v(x) & = \log \left( (1-\beta) \iota(x)^{-\frac{1}{\kappa}} + \beta \lambda^{\frac{1}{\kappa}} \tilde{\mb E}\left[ \left. \left(  (1-\beta)\sum_{n=0}^\infty (\beta \lambda^{\frac{1}{\kappa}})^n \tilde{\mb E}^n (\iota^{-\frac{1}{\kappa}})(X_{t+1}) \right)^{\kappa} \right| X_t = x \right]^{\frac{1}{\kappa}} \right) \\
 & \leq \log \left( (1-\beta) \iota(x)^{-\frac{1}{\kappa}} + \beta \lambda^{\frac{1}{\kappa}} \tilde{\mb E}\left[ \left. (1-\beta)\sum_{n=0}^\infty (\beta \lambda^{\frac{1}{\kappa}})^n \tilde{\mb E}^n (\iota^{-\frac{1}{\kappa}})(X_{t+1}) \right| X_t = x \right]\right) \\
 & = \log \left( (1-\beta) \iota(x)^{-\frac{1}{\kappa}} +  (1-\beta)\sum_{n=1}^\infty (\beta \lambda^{\frac{1}{\kappa}})^n \tilde{\mb E}^n (\iota^{-\frac{1}{\kappa}})(x) \right)  = \bar v(x) \,.
\end{align*}
Existence now follows by Proposition \ref{p:exun}(i).
\end{proof}

\begin{lemma}\label{lem:T-prop-ez}
Let condition (\ref{e:ez-phi-1}) hold. Then for any $\kappa \neq 0$, the operator $\mb T$ from (\ref{e:T-ez-1}) is a continuous, monotone operator on $\tilde E^{\phi_s}$ for each $1 \leq s \leq r$.
\end{lemma}

\begin{proof}[Proof of Lemma \ref{lem:T-prop-ez}]
Fix any $s \in [1,r]$. We first show that $\mb E^\mu[ e^{|\mb Tf(X_t)/c|^s} ]<\infty $ holds for any $f \in \tilde E^{\phi_s}$ and $c$ sufficiently small.  By convexity of $x \mapsto e^{|(\log x)/c|^s}$ for $c \in (0,1]$ and two applications of Jensen's inequality and iterated expectations, for any $c \in (0, 1 \wedge |\kappa|^{-1}]$ we obtain
\begin{align*}
 \mb E^{\tilde \mu} \left[ e^{|\mb Tf(X_t)/c|^s} \right]
 & = \mb E^{\tilde \mu} \left[ \exp \left( \left| \log \left( (1-\beta) \iota(X_t)^{-\frac{1}{\kappa}} + \beta \lambda^{\frac{1}{\kappa}} \tilde{\mb E}[ e^{\kappa f(X_{t+1})} | X_t]^{\frac{1}{\kappa}} \right)  /c \right|^s \right) \right] \\
 & \leq \mb E^{\tilde \mu} \left[ (1-\beta) e^{ \left| \log  \iota(X_t) / (\kappa c) \right|^s }  + \beta \exp \left( \left| \log \left(  \tilde{\mb E}[ \lambda e^{\kappa f(X_{t+1})} | X_t ] \right)  /(\kappa c) \right|^s \right) \right] \\
 & \leq \mb E^{\tilde \mu} \left[ (1-\beta) e^{ \left| \log  \iota(X_t) / (\kappa c) \right|^s } + \beta \tilde{\mb E}\left[ \left. \exp \left( \left| \log \left(   \lambda e^{\kappa f(X_{t+1})} \right)  /(\kappa c)  \right|^s  \right) \right| X_t \right]\right] \\
 & = (1-\beta) \mb E^{\tilde \mu} \left[  e^{ \left| \log  \iota(X_t) / (\kappa c) \right|^s } \right] + \beta \mb E^{\tilde \mu} \left[ e^{ \left| (\log \lambda) / (\kappa c) + f(X_t)  /c  \right|^s } \right] \,,
 \end{align*}
where the right-hand side is finite under condition (\ref{e:ez-phi-1}), and the final equality is because $\tilde \mu$ is the stationary distribution under $\tilde{\mb E}$. It follows by Remark \ref{rmk:r3} that $\mb T : \tilde E^{\phi_s} \to \tilde E^{\phi_s}$.

For continuity, fix $f \in \tilde E^{\phi_s}$ and take any $h \in \tilde E^{\phi_s}$ with $\|h\|_{\phi_s}$ (with the norm defined relative to the measure $\tilde \mu$) sufficiently small in a sense we make precise below. Then 
\[
 \mb T(f+h)(x) - \mb Tf(x)
 = \log \left\{ \frac{ (1-\beta) \iota(x)^{-\frac{1}{\kappa}} + \beta \lambda^{\frac{1}{\kappa}} w(x) \tilde{\mb E}_f[ e^{\kappa h(X_{t+1})} | X_t = x]^{\frac{1}{\kappa}} }{ (1-\beta) \iota(x)^{-\frac{1}{\kappa}} + \beta \lambda^{\frac{1}{\kappa}} w(x) }  \right\} 
\]
where $w(x) = \tilde{\mb E}\left[ \left. e^{\kappa f(X_{t+1}) }  \right| X_t = x \right]^{1/\kappa}$ and $\tilde{\mb E}_f$ denotes the distorted conditional expectation operator $\tilde{\mb E}_f g(x) := \tilde{\mb E}[ m_f(X_t,X_{t+1}) g(X_{t+1}) | X_t = x]$ where
\[
 m_f(X_t,X_{t+1}) = \frac{e^{ \kappa  f(X_{t+1}) }}{\tilde{\mb E}[ e^{ \kappa  f(X_{t+1})  } | X_t ]} \,.
\]
Take any $c \in (0,1 \wedge |\kappa|^{-1}]$. By convexity of $x \mapsto e^{|(\log x)/c|^s}$ for $c \in (0,1]$, two applications of Jensen's inequality, and the Cauchy--Schwarz inequality, we obtain
\begin{align*}
 & \mb E^{\tilde \mu} \left[ e^{|(\mb T(f+h)(X_t) - \mb Tf(X_t))/c|^s} \right] \\
 & = \mb E^{\tilde \mu} \left[ \exp \left( \left| \frac{1}{c} \log \left\{ \frac{(1-\beta)\iota(X_t)^{-\frac{1}{\kappa}} + \beta \lambda^{\frac{1}{\kappa}} w(X_t) \tilde{\mb E}_f \left[ \left. e^{ \kappa h(X_{t+1}) }  \right| X_t \right]^\frac{1}{\kappa} }{ (1-\beta) \iota(X_t)^{-\frac{1}{\kappa}} + \beta \lambda^{\frac{1}{\kappa}} w(X_t) }  \right\} \right|^s \right) \right] \\
 & \leq \mb E^{\tilde \mu} \left[ \frac{(1-\beta) \iota(X_t)^{-\frac{1}{\kappa}} + \beta \lambda^{\frac{1}{\kappa}} w(X_t) e^{ \left| \frac{1}{c\kappa} \log \tilde{\mb E}_f \left[ \left. e^{\kappa h(X_{t+1}) }  \right| X_t  \right] \right|^s } }{ (1-\beta) \iota(X_t)^{-\frac{1}{\kappa}} + \beta \lambda^{\frac{1}{\kappa}} w(X_t) }  \right] \\
 & \leq \mb E^{\tilde \mu} \left[ e^{ \left| \frac{1}{c\kappa} \log \tilde{\mb E}_f \left[ \left. e^{\kappa h(X_{t+1}) }  \right| X_t  \right] \right|^s } \right] \\
 & \leq \mb E^{\tilde \mu} \left[ \tilde{\mb E}_f \left[ \left.   e^{ \left|  h(X_{t+1}) /c  \right|^s } \right| X_t \right]  \right] \\
 & \leq \mb E^{\tilde \mu} \left[  e^{4|\kappa f(X_t)|} \right]^\frac{1}{2}  \mb E^{\tilde \mu} \left[  e^{ 2\left|   h(X_t)  /c\right|^s }   \right]^\frac{1}{2}\,.
\end{align*}
For $h \in E^{\phi_s}$ with $\|h\|_{\phi_s} \leq 2^{-1/s} (1 \wedge |\kappa|^{-1})$, setting $c = 2^{1/s}\|h\|_{\phi_s}$ we therefore have
\[
 \mb E^\mu\left[ e^{|(\mb T(f+h)(X_t) - \mb Tf(X_t))/(2\|h\|_{\phi_s})|^s} \right] \leq \left( 2 \mb E^{\tilde \mu} \left[  e^{4|\kappa f(X_t)|} \right] \right)^\frac{1}{2}  \,.
\]
Continuity now follows by Lemma \ref{lem:pollard}.
Monotonicity of $\mb T$ follows form monotonicity of conditional expectations and monotonicity of the $\log$ and $\exp$ functions.
\end{proof}

\begin{proof}[Proof of Corollary \ref{c:ez-exist-2}]
Immediate from Theorem \ref{t:ez-exist-1} and Lemma \ref{lem:embed:0}.
\end{proof}

\let\oldbibliography\thebibliography
\renewcommand{\thebibliography}[1]{\oldbibliography{#1}
\setlength{\itemsep}{2pt}}

{
\singlespacing
\putbib
}

\end{bibunit}

\end{document}